\documentclass[11pt,reqno]{amsart}
\usepackage{amsmath,amsthm,amssymb,amsfonts,bbm,mathtools,dsfont}
\usepackage{enumerate,hyperref,mathtools}
\usepackage{young}
\usepackage{graphicx}
\usepackage{float}
\usepackage[mathscr]{eucal}
\usepackage[a4paper]{geometry}
\usepackage{microtype}
\usepackage{hyperref}
\usepackage{pslatex}

\newtheorem{maintheorem}{Theorem}
\newtheorem{theorem}{Theorem}[section]
\newtheorem{corollary}[theorem]{Corollary}
\newtheorem{lemma}[theorem]{Lemma}

\newtheorem{claim}[theorem]{Claim}
\newtheorem{proposition}[theorem]{Proposition}
\newtheorem{definition}[theorem]{Definition}

\begin{document}

\title{Simulated Tempering and Swapping on Mean-field Models }

\author{Nayantara Bhatnagar}

\address{Department of Mathematical Sciences,  University of Delaware, Newark, DE, 19716.}
\email{nayantara.bhatnagar@gmail.com}

\author{Dana Randall}
\address{School of Computer Science and School of Mathematics, Georgia Institute of Technology, Atlanta, GA 30332.}
\email{randall@cc.gatech.edu}

%
%
%

\maketitle

\begin{abstract}

Simulated and parallel tempering are families of Markov Chain Monte
Carlo algorithms where a temperature parameter is varied during the
simulation to overcome bottlenecks to convergence due to multimodality.

In this work we introduce and analyze the convergence for a set of new
tempering distributions which we call {\em entropy dampening}. For asymmetric
exponential distributions and the mean field Ising model with and
external field simulated tempering is known to converge
slowly. We show that tempering with entropy dampening distributions
mixes in polynomial time for these models.

Examining slow mixing times of tempering more closely, we show that for the
mean-field 3-state ferromagnetic Potts model, tempering converges
slowly regardless of the temperature schedule chosen. On
the other hand, tempering with entropy dampening distributions
converges in polynomial time to stationarity.
Finally we show that the slow mixing can be very
expensive practically. In particular, the mixing time of simulated
tempering is an exponential factor longer than the mixing time at the
fixed temperature.

\end{abstract}

\section{Introduction}
Markov Chain Monte Carlo (MCMC) methods for sampling from a target
distribution have become ubiquitous in Bayesian statistics
\cite{RobCas}, fields such
as machine learning \cite{AdDJ}, and in simulations of large physical
systems \cite{RubKro}. MCMC has also played an important role in several central
results in the theory of algorithms \cite{JSV,LV}.

It is usually straightforward to design an MCMC algorithm
which converges to the desired target distribution. Unfortunately, a common
difficulty in applications from statistics and statistical physics is
multimodality in the distribution which can cause the algorithm to take
an impractical amount of time to converge.

{\em Simulated tempering} \cite{MarPar,GeyTho} and Metropolis-coupled
MCMC (or {\em
  parallel tempering}, or {\em swapping}) \cite{Gey} are Markov chain
samplers related to
simulated annealing which are widely used for sampling in the presense of
multimodality. Their popularity in practice makes it important to
understand their convergence rates theoretically.  In these algorithms,
a temperature
parameter is randomly updated over a range of values during the
simulation. The idea is to speed up sampling at low temperatures,
circumventing the bottlenecks of the multimodal distribution by
sampling some of the time at higher temperatures. In the following
discussion, by ``converges quickly'' or that there is ``fast convergence" we mean that the Markov chain
converges to within a small distance of the equilibrium stationary
distribution in time
that is polynomial in the size of the states. ``Slow convergence'' means that
even after an exponential amount of time in this parameter, the chain is far from the
equilibrium distribution.


Obstructions to the fast
convergence of the dynamics often arise in models from statistical physics which exhibit phase transitions.
The {\em Ising model} and its generalization, the {\em Potts model} are models
from statistical physics of large numbers of interacting particles
where the equilibrium distribution exhibits multimodality. Due to
this, local MCMC algorithms for simulating such systems converge slowly.
Madras and Zheng \cite{MZ} analyzed simulated tempering
for two symmetric bimodal distributions - the mean field Ising model
and an ``exponential valley'' distribution. In both cases they
showed that tempering and swapping converge quickly at any
temperature. Their analysis makes use of the decomposition theorems of
\cite{MR} which say that it is sufficient to bound the convergence time
of the chain within each mode and of the macro-chain over all the
temperatures by a polynomial.

In contrast, in \cite{BR}, we showed that for the 3-state mean-field
ferromagnetic Potts model (which is not bimodal, but rather, has three
modes at low temperature), simulated tempering converges prohibitively
slowly. This is caused by a 
phase transition in the Potts model which is of a different type than the
phase transition in the Ising model. In
fact due to the nature of the phase transition, simulated tempering
converges slowly regardless of the intermediate temperatures chosen
for tempering\footnote{Some of these results appear in an extended abstract \cite{BR} and in
  thesis form \cite{thesis}. This is the full version which contains complete proofs of all the results.}. The proof of this theorem appears in unpublished form in \cite{thesis}  and we include the full proof here.

Woodard, Huber and Schmidler \cite{Woodard,Woodard2}
generalized the above examples to give frameworks for polynomial
and slow convergence of simulated and parallel tempering for more general
measures and state spaces. In particular in \cite{Woodard}, they give sufficient
conditions on a distribution and a sequence of temperatures for
simulated and parallel tempering to converge polynomially. In \cite{Woodard2} they
show several cases where if the above conditions are violated,
simulated and parallel tempering will converge slowly. In particular,
one property that plays an important role is what they term the ``persistence'' of a
distribution. Roughly, they show that if there is a single mode of the
distribution which is very narrow or ``spiky'' compared to the other
modes but has about the same probability mass then tempering  and
swapping converge slowly. Woodard et al. use this property to explain
the slow convergence for the 3-state Potts model as well.

The main results of the current paper touch upon these last points.
In \cite{BR}, we also extended the results of
Madras and Zheng showing polynomial convergence of swapping for symmetric
exponential distributions to the case of an assymetric exponential
distribution. This more general result leads to the insight that it is
possible to choose the distributions for tempering more advantageously
if we do not restrict to distributions that are paramterized by temperature. As an
application in \cite{BR} we cited polynomial convergence of swapping (and
hence simulated tempering) for the mean-field Ising model with an
external field. In this paper we present the full proof of this
result. We define certain ``entropy dampening distributions''
which make use of properties of the stationary measure. We show
that if entropy dampening distributions are used for the 3-state
mean-field Potts model, then in fact tempering mixes polynomially. These
examples of polynomial mixing do not fall under the sufficient conditions
given in \cite{Woodard} since we make use of more general
tempered densities. 

Lastly, we show that there are cases when the mixing
time of the tempering algorithm can be significantly slower than that
of the fixed temperature Metropolis algorithm; i.e., even if we use a
polynomial number of distributions, the mixing time of the tempering chain
may be exponentially larger than that of the chain at the fixed low
temperature. This contradicts the conventional
wisdom that tempering can be in the worst case slower by a factor that is
polynomial in the number of temperatures. Our proof makes use of
sharper results about the slow mixing beyond the conditions in
\cite{BR,Woodard}.

We point out that the examples considered here are tractable by other
means and one does not need a Markov
chain to sample configurations. Neverthless, we feel the methods
presented here offer some insight
into how to design more robust tempering algorithms in general.

The remainder of this paper is organized as follows. In Section \ref{sec:spinsys} we
present some preliminaries on spin systems and Markov chains. In
Section \ref{sec:sim-temp-swap} we define the simulated tempering and swapping
algorithms formally. The statements of the main theorems can be found in Section \ref{sec:mainthms}.
In Section \ref{fast} we analyze the convergence time of the swapping
algorithm for asymmetric exponential distributions. In Section
\ref{newswap} we show that the swapping algorithm using a modified
entropy dampening distribution mixes polynomially for sampling from the
mean-field Ising model with any external field.
In Section \ref{torpid} we show that
there is a temperature at which simulated tempering mixes
exponentially slowly for the 3-state
mean-field ferromagnetic Potts model. Finally, in Section \ref{slower}
we show that in certain cases tempering can slow down the convergence
of the Metropolis algorithm by an exponential factor.

\section{Spin Systems and Markov Chains}
\label{sec:spinsys}
The {\em $q$-state Potts model} \cite{potts} on a finite graph $G=(V,E)$ at
inverse temperature $\beta \ge 0$ with an external magnetic field is
defined as follows. The set of possible {\em configurations} $\Omega$ is
$\{1,\ldots,q\}^{V}$ where a configuration $x$ is an assignment of one
of $q$ {\em spins} to each vertex of $G$ and $x_i$ denotes the spin of $i \in V$. The case $q=2$ corresponds to
the classical special case of the {\em Ising model}.
In this case, the set of spins is conventionally taken to be $\{+1,-1\}$ and we will follow this notation. Spins may also be referred to as colors and we use the two interchangeably. For a spin configuration $x$, let $\sigma(x) = \sigma = (\sigma_1,\ldots,\sigma_q)$ where $\sigma_i$ denotes the number of vertices on $x$ with spin $i$ for $1 \le i \le q$. The {\em Hamiltonian} of a configuration $x$ is
defined by
\begin{align*}
H(x) = \sum_{(i,j) \in E}\delta_{x_i,x_j}+\sum_{m=1}^q\sum_{i \in V}h_m\delta_{x_i,m}
\end{align*}
where $\delta_{x_i,x_j}$ is the Kronecker delta function and
$h=\{h_m\}_{m=1}^q$ are real numbers representing the external fields.
The probability that the system has a given configuration $x$ at
inverse temperature $\beta$ is given by the {\em Gibbs distribution}:
\begin{eqnarray*}
\pi_{{\beta}}(x) = \frac{e^{{\beta} H(x)}}{Z(\beta,h)},
\end{eqnarray*}
where $Z({\beta,h})$ is a normalizing constant known as the {\em
  partition function}. In the case that the external fields are 0, we denote the partition function by $Z(\beta)$. The higher the inverse temperature $\beta$ the
more the distribution favors configurations which have many
neighboring vertices with the same spin. At $\beta=0$, i.e. infinite
temperature, the Gibbs distribution is uniform over all
configurations. Here we are concerned with $\beta \ge 0$ which is the
{\em ferromagnetic} Potts model. In contrast, in the {\em anti-ferromagnetic} case $\beta<0$, neighbors in the underlying
graph prefer to have different spins. 

\subsection{Markov Chains}
In the MCMC method, the Markov chain performs
a random walk on the {\em Markov kernel}, which is a graph defined on
the space of configurations. One such Markov chain for sampling from a
Gibbs distribution is the {\em heat
  bath Glauber dynamics}. Starting at a state $x_0 \in \Omega$, at each
time step a vertex is chosen at random from $V$ and its spin is updated by
choosing it according to $\pi_\beta$ conditioned on the spins of the
other vertices. Thus the kernel for this chain is the graph on
$\Omega$ where there is an edge between two configurations if they
differ by the spin of one vertex. It can be checked that the chain is
reversible with respect to $\pi_\beta$ and ergodic and thus $\pi_\beta$ is the stationary
distribution.

In general, given a connected kernel, it is straightforward to sample from a desired
distribution $\pi$ on $\Omega$ using the Metropolis-Hastings algorithm \cite{met}. Suppose
that $Q$ is the transition matrix of irreducible, symmetric Markov
chain over the state space $\Omega$; this will be the {\em proposal
  chain}. The transition matrix $P$ of the Metropolis Markov chain is
given by
\begin{align*}P(x,y) = \left\{
\begin{array}{ll}
Q(x,y) \min\left( 1, \frac{\pi(y)}{\pi(x)} \right) & \textrm{ if } y
\neq x \\
1-\sum_{z \neq x}P(x,z) & \textrm{ if } y = x.
\end{array}
\right.
\end{align*}
The chain $P$ is irreducible and reversible with respect to $\pi$ and therefore $\pi$ is the stationary distribution of $P$ (see e.g. \cite{DiaSal98}).

The convergence of a Markov chain can be measured by the {\em mixing time}, the
time for the chain to come within a small distance of the equilibrium
distribution. Let $(X_t)$ be an ergodic (i.e., irreducible and
aperiodic), reversible Markov chain with finite state space $\Omega$,
transition probability matrix $P$, and stationary distribution $\pi$.
Let $P^t(x,y)$ denote the $t$-step transition probability from $x$ to $y$.
\begin{definition}
The {\em total variation distance} at time $t$ of $(X_t)$ to
stationarity is
\[
\|P^t,\pi\|_{tv}=\max_{x\in\Omega}\frac{1}{2}\sum_{y\in\Omega}|P^t(x,y)-\pi(y)|.
\]
\end{definition}
\begin{definition}
Let $0<\varepsilon<1$, then the {\em mixing time} $\tau(\varepsilon)$
is defined to be
\[
\tau(\varepsilon) := \min\{t : \|P^{t'},\pi\|_{tv} \leq\varepsilon,
       \forall t'\geq t\}.
\]
\end{definition}
We say that the Markov chain $(X_t)$ {\em mixes polynomially} if
the mixing time is bounded above by a polynomial in $n$ and
$\log\frac{1}{\varepsilon}$, where $n$ is the number of coordinates
of each configuration in the state space.  When the mixing time is exponential
in $n$, we say the chain {\em mixes torpidly} or {\em slowly} or {\em exponentially slowly}.

There are several methods to obtain a bound on the mixing time of a
Markov chain. The inverse of the {\em spectral gap} of the transition matrix
of a Markov chain characterizes the mixing time as follows.
Let $\lambda_0,\lambda_1,\ldots,\lambda_{|\Omega|-1}$
be the eigenvalues of an ergodic reversible Markov chain with
transition matrix $P$, so that
$1=\lambda_0 > |\lambda_1| \geq |\lambda_i|$ for all $i\geq 2$.
Let the spectral gap be $Gap(P): = \lambda_0 - |\lambda_1|$.

\begin{theorem}[\cite{LevPerWil09}]\label{thm:spec_gap_thm}
For any $\varepsilon>0$,
\[  \left(\frac{1}{Gap(P)}-1 \right) \log\left(\frac{1}{2\varepsilon}\right)     \le \tau(\varepsilon)
  \leq \frac{1}{Gap(P)}\log \left(\frac{1}{\pi^*\varepsilon} \right) \]
\end{theorem}
where $\pi^*=\min_{x} \pi(x)$.

\subsection{Mean-field Models}

The {\em Curie-Weiss} or {\em mean-field} Potts model corresponds to
the case when the graph $G$ is the complete graph on $n$
  vertices. Mean-field models are studied (see e.g. \cite{BisCha} and references therein) because
  often in higher dimensions they share
  characteristics of the model on lattices.

For the mean-field Potts at low enough temperatures, local dynamics such as
Glauber dynamics mix exponentially slowly
\cite{bcfktvv,DFJ,Mar,Ran,Tho}. This is because at low temperature,
the distribution is multimodal, consisting of ordered modes
corresponding to configurations which are predominantly of one
spin. These modes are separated by configurations which are
exponentially unlikely in the Gibbs distribution. As the temperature
is raised, there is a critical temperature beyond which a single mode of
disordered configurations dominates since the contribution of the
entropy of configurations dominates the energy, or
Hamiltonian, term. For more details on the mixing time of Glauber dynamics for
mean-field models, see \cite{DLP,CDLLPS}.

The Swendsen-Wang (SW) algorithm \cite{SweWan} is another algorithm proposed as an
alternative to local dynamics for sampling from configurations of
the $q$-state Potts model. Cooper, Dyer, Frieze and Rue \cite{CDFR} considered the
mean-field Ising model and showed that the SW algorithm mixes
polynomially at
all temperatures except possibly near the critical point. Gore and
Jerrum \cite{GJ} showed that the SW algorithm mixes
torpidly on the mean-field Potts model for $q \ge 3$ at the critical
temperature. Long, Nachmias and Peres \cite{LongNachmiasPeres} have resolved
the order of the mixing time of SW at the critical point for the Ising model. Recently, Galanis, Stefankovic and Vigoda and have studied the mixing time of the Swendsen-Wang algorithm for the mean-field Potts model when $q \ge 3$ and shown four different regimes depending on the inverse temperature \cite{GalSteVig15}.

\section{Simulated Tempering and Swapping}\label{sec:sim-temp-swap}

Simulated and parallel tempering are families of Markov chain
algorithms that have been proposed for sampling from multimodal
distributions. They are used widely in practice and their convergence
behavior for mean-field models has led to a better understanding of when
these algorithms can speed up mixing of local Markov chains
\cite{MZ,Z,BR,Woodard,Woodard2}.
The simulated and parallel tempering Markov chains are built on top of
a fixed temperature Metropolis-Hastings Markov chain. We define these
chains in the context of sampling Gibbs distributions below, although it will be clear from the definitions that the chains may be defined for more general distributions, each on the space $\Omega$.

\subsection{Simulated tempering. }\label{sec:simtemp}
Suppose that we wish to sample from a Gibbs distribution
$\pi_\beta$ over $\Omega$ at
inverse temperature $\beta$. The simulated
tempering Markov chain is defined as follows \cite{MarPar,GeyTho}.
Fix $0 = \beta_0 <  \ldots < \beta_M=\beta$, a sequence of inverse
temperatures.
The state space of the simulated tempering chain is given by
\[
\Omega_{st} = \Omega \times \{0,\ldots,M\}.
\]
Define the $i^{th}$ tempering distribution $\pi_i$ as the Gibbs
distribution at $\beta_i$
\[\pi_i := \pi_{\beta_i}, \ \ 0 \leq i \leq M. \]
Denote by $M_i$ a Metropolis-Hastings chain for sampling from
$\pi_{i}$ where at each time, independently, the proposal chain chooses a vertex in $V$ where $|V|=n$ and spin in $\{1,\ldots,q\}$ independently and uniformly at random.

The tempering Markov chain consists of two types of transitions:
{\it level moves}, which
update the configuration while keeping the temperature fixed,
and {\it temperature moves,} which update the temperature
while remaining at the same configuration. In each step of the chain,
we randomly choose with equal probability one of the two transitions to
perform (c.f. \cite{MZ} for other ways to define the chain). 

\vskip.1in
\noindent $\bullet$ A \underbar{\bf level move} connects $(x,i)$ and
$(x',i)$ with the transition probability given by
\[
P_{st}((x,i),(x',i)):= \frac{1}{2}M_i(x,x').
\]

\vskip.1in
\noindent $\bullet$ A \underbar{\bf temperature move}
connects $(x,i)$ to $(x,i\pm1)$. If the current state is $(x,i)$, choose an inverse temperature $j = i\pm 1$ with probability
$r_{i,j}$, where $r_{0,1}=r_{M,M-1}=1/2$ and $r_{i,j} =1/2$ for
$0<i<M$. The move to $(x,j)$ is accepted with the appropriate Metropolis
probability. Thus, the transition probabilities are given by
\begin{align*}
P_{st}((x,i),(x,j)):=
\begin{cases}
\frac{1}{2}r_{i,j} \min
\left(1,\frac{\pi_j(x)}{\pi_i(x)}\right) & |j-i|=1\\
0 & |j-i|>1\\
\frac 12- \displaystyle\sum_{j=\pm 1} \frac{1}{2} r_{i,j} \min
\left(1,\frac{\pi_j(x)}{\pi_i(x)}\right) & j=i
\end{cases}
\end{align*}

It is straightforward to verify that the chain is reversible with respect to $\pi_{st}$. The transition probabilities ensure that the stationary distribution $\pi_{st}$ is uniform over all temperatures
and the conditional distributions $\pi_{st}(\cdot,i)$ for $0 \le i \le M$ are proportional to the fixed
temperature Gibbs distributions $\pi_i$. That is,
\[
\pi_{st}(\cdot,i) = \frac{1}{M+1} \pi_i(\cdot).
\]

If $M$ is chosen to be a
polynomial in $n$, the stationary weight of
set of states at each fixed inverse temperature $\beta_i$ is at least an inverse polynomial
fraction of the state space. A common choice of inverse temperatures is to take $\beta_i = i\beta/M$. It can be
verified that in this case if $M$ is at least polynomial in $n$, the transition
probabilities are non-negligible, by bounding the size of the ratio
$\frac{Z(\beta_i)}{Z(\beta_{i\pm1})}$.
Notice that while the exponential factor
is simple to calculate given $x$ and $i$,
it is not clear that we can compute the ratio of partition functions in order to implement the simulated tempering algorithm.
The swapping algorithm is designed to avoid this difficulty in
implementing temperature moves.

\subsection{Swapping. } The swapping algorithm was defined by Geyer
\cite{Gey}. Let $\pi_\beta$ be the Gibbs distribution from which we wish
to sample. Fix $0 = \beta_0 <  \ldots < \beta_M=\beta$, a sequence of inverse
temperatures. The state space is the product space $\Omega_{sw} =
\Omega^{(M+1)}$, the product of $M + 1$
copies of the original
state space, where each coordinate corresponds to an inverse temperature.
A configuration in the swapping chain
is denoted by an $(M + 1)$-tuple ${x}=(x_0,\ldots,x_M)$.
As before, define $\pi_i := \pi_{\beta_i}$ for $0 \leq i \leq M$ and
let $M_i$ be a Metropolis-Hastings chain for sampling from
$\pi_{i}$ where at each time proposal chain chooses a vertex and spin
independently and uniformly at random.
Define $\pi_{sw}$ to be the product measure of the distributions $\pi_i$
\[\pi_{sw}(x) := \prod_{i=0}^M \pi_i(x_i).\]
In each step, the swapping Markov chain chooses an inverse temperature $\beta_i$
uniformly at random and chooses uniformly from the following two types of transitions.

\vskip.1in
\noindent $\bullet$ A \underbar{\bf level move} connects
${x}=(x_0,\ldots,x_i,\ldots,x_M)$
and ${x'}=(x_0,\ldots,x_i',\ldots,x_M)$ if ${x}$ and
${x'}$ agree in all but the $i^{th}$ components, and $x_i$ and $x_i'$
are connected by one-step transitions of the Metropolis algorithm on $\Omega$.
In this case, the transition from $x$ to $x'$ has transition probability
\[
P_{sw}(x,x') = \frac{1}{2(M\!+\!1)}M_i(x,x').
\]
\noindent $\bullet$ A \underbar{\bf swap move} connects
${x}=(x_0,\ldots,x_i,x_{i+1},\ldots,x_M)$ to
${x'}=(x_0,\ldots,x_{i+1},x_i,\ldots,x_M)$, i.e., it exchanges
the $i^{th}$ and ${i+1}^{st}$ components with an appropriately chosen
Metropolis probability. From the current state $x$, choose a
coordinate $i$ uniformly at random from $\{0,\ldots,M-1\}$. Let $x'$ be the configuration obtained by exchanging the $i^{th}$ and ${i+1}^{st}$ components of $x$. Then, the probability of the transition from $x$ to $x'$ is given by 
\begin{eqnarray*}
P_{sw}(x,x') & = & \frac{1}{2M}
\min\left(1, \frac{\pi_{sw}({x'})}
{\pi_{sw}({x})}\right) \\
& = & \frac{1}{2M} \min\left(1, \frac{\pi_{i+1}(x_i)
\pi_i(x_{i+1})}{\pi_i(x_i)\pi_{i+1}
(x_{i+1})}\right) \\
& = & \frac1{2M} \min
\left(1, e^{(\beta_{i+1}-\beta_{i})(H(x_i)-H(x_{i+1})}\right).
\end{eqnarray*}
It can be verified that the chain is reversible with respect to $\pi_{sw}$ and thus has stationary distribution $\pi_{sw}$. Since $\pi_{sw}$ is a product measure, samples according to $\pi_M$ can be obtained
by projecting on the last co-ordinate. Notice that in the transition probabilities above, the normalizing constants cancel out.  Hence,
implementing a move of the swapping chain is straightforward, unlike
tempering where good approximations for the partition functions are required.
Zheng proved that fast mixing of the swapping chain
implies fast mixing of the tempering chain \cite{Z}. The converse result is not known.

To define the tempering and swapping Markov chains in the case of the mean-field models that we will study, the base proposal chain for the Metropolis chain at a fixed temperature will be the {\em heat bath Glauber dynamics}. That is, at each time step, a uniformly random vertex and a uniformly random spin is chosen, and the spin of the chosen vertex is updated to the chosen spin.

For both tempering and swapping, we must be careful about how we
choose the number of distributions $M+1$ and the distributions themselves.  It is important that successive
distributions $\pi_i$ and $\pi_{i+1}$ have sufficiently small
variation distance so that temperature moves are accepted with
non-trivial probability.
At the same time, $M$ must be small enough so that 
running time of
the algorithm does not become very large. Setting $M$ to be a polynomial which is $\Omega(n)$ and setting $\beta_i = i\beta/M$ is often a reasonable choice.

In general, one can define a
sequence of distributions $\pi_0,\ldots,\pi_M$ and define the simulated
tempering and swapping chains with these as the fixed temperature
distributions by defining the transition probabilities using the
Metropolis rule as before. In the sequel we will make use of this to
define tempering and swapping chains for the mean-field Ising model
with an external field.

\section{Results for Mean-Field Models}
\label{sec:mainthms}
Although the simulated
tempering and swapping algorithms have been defined above as having fixed
temperature distributions which are Gibbs distributions, in fact the
algorithms are more general. Our first result makes use of this and
bounds the mixing time of the swapping chain which uses a different
set of fixed temperature distributions (defined in Section
\ref{sec:ent-damp}).
\begin{maintheorem}
The swapping Markov chain for the ferromagnetic mean-field Ising model
using {\em entropy dampening
distributions} mixes polynomially for every inverse temperature $\beta>0$
and any external field.
\end{maintheorem}

Unlike the Ising model, simulated tempering for the 3-state Potts model mixes slowly.

\begin{maintheorem}\label{thm:slow-potts}
Let $\beta_c = \frac{4\ln2}{n}$. There is a constant $c_1 > 0$ such that for any set of inverse temperatures $\beta_c = \beta_M \ge \cdots \ge \beta_0 \ge 0$ such that $M = n^{O(1)}$, the
tempering and swapping chains with the distributions $\pi_{\beta_i}$ for the $3$-state mean-field ferromagnetic Potts model have mixing time
$\tau(\varepsilon) \geq
e^{c_1n}\ln(1/\varepsilon)$.
\end{maintheorem}

The slow convergence of the tempering chain is caused by a {\em first order
phase transition} in the 3-state
ferromagnetic Potts model. First order phase transitions are
characterized by phase-coexistence of ordered and disordered phases at a critical temperature \cite{Geo}. In contrast,
the Ising model has a second-order (continuous) phase transition, and
there is no
phase coexistence, and this distinguishes why
simulated tempering works for one model and not the other. Our
techniques using entropy dampening distributions do not seem to extend
immediately to show polynomial mixing of the swapping algorithm for the Potts model.

Let $\Omega_{RGB}$
denote the subset of the state space of the $3$-state Potts model
$\Omega$ where $\sigma_1 \geq \sigma_2 \geq \sigma_3$. On the restricted space
$\Omega_{RGB}$, we show that tempering can
slow down the Metropolis algorithm at a fixed temperature by an
exponential multiplicative factor.

\begin{maintheorem}\label{thm:temp-vs-metropolis}
There are constants $c,c'$ with $0 < c' < c$ and an inverse temperature $\beta$ 
such that the
Metropolis chain on $\Omega_{RGB}$ at $\beta$ has mixing time
$\tau(\varepsilon) \leq e^{c'n}\ln(1/\varepsilon)$ while the mixing
time of the tempering chain is bounded by $\tau(\varepsilon) \geq
e^{cn}\ln(1/\varepsilon)$.
\end{maintheorem}

Though the mixing time of the Metropolis chain is exponential, to
obtain this upper bound, it is not sufficient to bound the conductance,
since such a bound is tight only up to quadratic factors. Instead, we
will appeal to a refinement of the comparison theorem of Diaconis and
Saloff-Coste.

\section{Swapping for the Asymmetric Exponential Distribution}
\label{fast}

In this section we show bounds on the mixing time of the swapping
Markov chain on the asymmetric exponential distributions
generalizing the symmetric exponential distribution for which Madras and
Zheng showed swapping mixes in polynomial time \cite{MZ}. This example will also serve as a warm-up for the
analysis of the next section for the mean-field Ising model.
While we focus
here on the swapping algorithm, which is easily implementable,
the distributions we define can also be used for tempering. Indeed, Zheng has shown that polynomial mixing of
swapping with any distributions implies polynomial mixing of tempering with
the same distributions \cite{Z}.

\subsection{Preliminaries}

The proof makes use of a Markov chain decomposition theorem
\cite{MR2,MarR}. Let $\mathfrak{M}$ be a Markov chain with transition
matrix $P$ and stationary distribution $\pi$. Let $\Omega_1,\ldots,\Omega_m$ be a disjoint partition of the state
space $\Omega$.
For each $i \in [m]$, define the Markov chain $\mathfrak{M}_{i}$ on
$\Omega_i$ whose transition matrix $P_i$, called the {\it
  restriction} of $P$ to
$\Omega_i$ is defined as
\begin{itemize}
\item $P_i(x,y)=P(x,y),$  \ \ \ if $x\ne y$ and $x,y\in \Omega_i$;
\item $P_i(x,x)=1-\displaystyle\sum_{y\in \Omega_i, y\ne x}
P_i(x,y),$  \  \ \ $\forall x\in \Omega_i$.
\end{itemize}
The stationary distribution of $\mathfrak{M}_{i}$ is given by $$\pi_i(x)\!=\!
\frac{\pi(x)}{\pi(\Omega_i)}, \ \ \ \  x \in \Omega_i.$$
Define the {\em projection} $\overline{P}$ to be the transition matrix
on the state space $[m]$
$$\overline{P}(i,j)=
     \frac{1}{\pi(\Omega_i)}\displaystyle\sum_{x\in\Omega_i,y\in\Omega_j}
\pi(x)P(x,y).$$
\vspace{-.1in}

The decomposition theorem bounds the spectral gap of the chain
$\mathfrak{M}$ by the spectral gap of the
projection chain and the gap of the slowest restriction chain.

\begin{theorem}[Martin and Randall \cite{MarR}]\label{dcmp}
$$ Gap(P) \ge \frac{1}{2}Gap(\overline{P})\left(\displaystyle\min_{i
    \in [m]}Gap(P_i)\right).$$
\end{theorem}

We will use a comparison theorem of Markov chains to bound the mixing
time of the projection chain defined below. The following comparison
theorem of Diaconis and
Saloff-Coste can be used to bound the mixing time of a Markov
chain when the mixing time of a related chain on the same space, but
with possibly a different stationary distribution is known.  Let
$\mathfrak{M}_1$ and $\mathfrak{M}_2$ be two Markov chains on
$\Omega$. Let $P_1$ and $\pi_1$ be the transition matrix and
stationary distributions of $\mathfrak{M}_1$ and  let $P_2$ and
$\pi_2$ be those of $\mathfrak{M}_2$. Let $E(P_1) =
\{(x,y) : P_1(x,y)>0\}$ and $E(P_2) = \{(x,y) :
P_2(x,y)>0\}$ be sets of directed edges. For $x,y \in
\Omega$ such that $P_2(x,y)>0$, define a {\it path}
$\gamma_{xy}$, a sequence of states $x=x_0,\ldots,
x_k=y$ such that $P_1(x_i,x_{i+1})>0$. 
Finally, let $\Gamma(z,w) = \left\{(x,y)\in E(P_2) :  (z,w)\in
\gamma_{xy}\right\}$ denote the set of endpoints of paths that use
the edge $(z,w)$.
\begin{theorem}[Diaconis and Saloff-Coste
  \cite{DS-C})]\label{thm:refined-compthm}
Let $a=\displaystyle\min_x\left(\frac{\pi_2(x)}{\pi_1(x)}\right).$
Then
\[ Gap(P_1) \geq \frac{a}{A}\cdot Gap(P_2), \]
where
$$A = \max_{(z,w)\in E(P_1)} \left\{{{1}\over{\pi_1(z) P_1(z,w)}}
\sum_{\Gamma(z,w)} |\gamma_{xy}|  \pi_2(x)
P_2(x,y)\right\}.$$
\end{theorem}
Note that in the case that the stationary distributions of the two
chains are the same, the above reduces to the more commonly used
version of the comparison theorem where $a=1$.

\subsection{The Bimodal Exponential Distribution}
Let $C >1$ be a real constant. Let $ N$ and  $N'$ be positive integers.
Define the bimodal exponential distribution over the integers in $[-N,N']$ as
\begin{eqnarray*}
\pi(x):= {{C^{|x|}}\over{Z }},   \ \ \ \ \ \ \ \ \  \ x \in  \{-N,\ldots,N'\},
\end{eqnarray*}
where $Z$ is the normalizing constant.
Define the distributions for the swapping
chain $P_{sw}$ as
\begin{eqnarray*}
\pi_i(x) \ := \ \frac{C^{\frac{i}{M}|x|}}{Z_i}, \ \ 0 \le i \le M, \ x
\in   \{-N,\ldots,N'\}
\end{eqnarray*}
where $Z_i$ is a normalizing constant and $M$ is the number of
distributions which we will assume is a polynomial in $N+N'$.
Let $P_i$ be the Metropolis-Hastings chain for sampling from $\pi_i$ where the base proposal chain is the simple symmetric random walk on $\{-N,\ldots,N'\}$. That is
\begin{eqnarray*}Q(i,j) = \left\{
\begin{array}{ll}
\frac 12 & \textrm{ if } |i-j| = 1 \textrm{ or } i=j=N \textrm{ or } N'. \\
0 & \textrm{ otherwise } 
\end{array}
\right.
\end{eqnarray*}
The state space for the swapping chain is
$\Omega_{sw}= \{-N,\dots,N'\}^{M+1}$ and its stationary distribution
is the product measure $\pi_{sw}$ of the distributions $\pi_i$.
\begin{theorem}
\label{expfast}
The swapping chain $P_{sw}$ with distributions $\pi_0,\ldots,\pi_M$ is
polynomially mixing.
\end{theorem}
Since our goal in this work is not to optimize the running times but
rather to distinguish between models where the mixing of tempering
and swapping are polynomial vs. exponential, we do not state the
precise polynomial upper bounds on the mixing time in the theorem.

\begin{definition} Let ${x}=(x_0,\ldots,x_M)
\in \Omega_{sw}$. The {\rm trace} Tr$({x}) = t:=(t_0,\ldots,t_M)
\in \{0,1\}^{M+1}$ where $t_i = 0$ if $x_i < 0$ and $t_i = 1$ if
$x_i \geq 0$, $i = 0,\ldots,M$.
\end{definition}
The $2^{M+1}$ possible values of the trace characterize the partition
we use to apply the decomposition theorem.  Letting
$\Omega_{sw}^t$ be the set of configurations with
trace $t$, we have the decomposition
\[\Omega_{sw} = \bigcup_{{t} \in \{0,1\}^{M+1}} \Omega_{sw}^t.\]

Let $P_t$ be the restriction of $P_{sw}$ to $\Omega^t_{sw}$, the configurations of fixed
trace $t$.  The state space of the projection $\overline{P}$ is the
$M+1$-dimensional
hypercube, representing the set of possible traces $t$. In \cite{MZ}, the spectral gap for a different version of the swapping chain for the symmetric exponential distribution was analyzed by decomposition. Our analysis of the restriction chains is similar to the analysis in \cite{MZ} for the correspondingly defined restriction chains. Analyzing the projection, however,
becomes more difficult, since in this case the stationary distribution
over the hypercube is highly non-uniform. This reflects the fact that
at ``low temperatures,'' one side of the bimodal distribution
becomes exponentially more favorable.  We will resolve this by an application of the comparison theorem with an auxilliary chain.

\vskip.1in
\noindent \underbar{\bf Mixing time of the restricted chains:}
\label{exisec}
\vskip.05in
\noindent
Note that if we ignore swap moves in the restriction chains $P_t$, the moves at each of the $M+1$ temperatures are independent and according to the Metropolis
probabilities. Let $\hat P_t$ be a modified chain which suppresses swap moves in $P_t$ while the trace is fixed at $t$.
The following lemma allows us to express the spectral gap of $\hat P_t$ in terms of the spectral gaps of the independent chains at each fixed temperature.
\begin{lemma}{\bf(Diaconis and Saloff-Coste \cite{DS-C})}\label{lem:product}
For $i=1,\ldots,m$, let $P_i$ be a reversible Markov chain on a finite
state space ${\Omega}_i$.  Consider the product Markov chain $P$ on the
product space $\Omega_0 \times \ldots \times \Omega_M$, defined by
\[
 P \,=\, \frac{1}{M+1}\sum_{i=0}^M
    \underbrace{I\otimes\ldots\otimes I}_{i} \otimes P_i
    \otimes \underbrace{I\otimes\ldots\otimes I}_{M-i}.
\]
Then  $Gap(P) \ = \ \frac{1}{M+1}\displaystyle\min_{0\leq i \leq
  M}\left\{Gap(P_i)\right\}\,.$
\end{lemma}
The distribution over $\Omega_{sw}^t$ restricted to each of the $M+1$
temperatures is unimodal, suggesting that $\hat P_t$
should be polynomially mixing at each temperature.
Madras and Zheng formalize this in \cite{MZ} and show that the
Metropolis chain restricted to the positive or negative parts of
$\Omega_i$ mixes polynomially. By Lemma \ref{lem:product}
and following the arguments as in \cite{MZ},
it can be shown that the Markov chains $\hat {P}_{{t}}$ are polynomially mixing for each $t \in \{0,1\}^{M+1}$. We omit these calculations here since they are exactly along the lines of those in \cite{MZ}. Next we show $P_t$ mixes in polynomial time by comparing it with $\hat P_t$.

\begin{lemma}\label{lem:poly-mixing-fixed-trace}
For each trace $t \in \{0,1\}^{M+1}$, the restriction chain $P_t$ mixes in polynomial time.
\end{lemma}
\begin{proof}

We make use of the comparison theorem Theorem \ref{thm:refined-compthm}. To apply the theorem, for each transition of $P_t$, we construct a canonical path consisting of moves in $\hat P_t$.

Let $(x,x')$ be a transition of $P_t$ with $x = (x_0,\ldots,x_M)$ and $x' = (x'_0,\ldots,x'_M)$. If $(x,x')$ is a level move which updates the state at a fixed temperature, let the corresponding canonical path in $\hat P_t$ be the edge $(x,x')$ itself. 
On the other hand, suppose $(x,x')$ is a swap move which exchanges the $i$th and $i+1$st components. Note that since the trace remains fixed, it must be the case that $t_i=t_{i+1}$. Suppose that $x_i \ge x_{i+1}\ge 0$. In this case define the canonical path as the concatenation of two paths $p_1 \circ p_2$, each consisting of a sequence of level moves at a fixed temperature.

\begin{itemize}
\item The path $p_1$ consists of $x_{i} - x_{i+1}$ level moves at the temperature $i+1$ from $x$ to \\$(x_0 ,\ldots x_{i-1}, x_{i} , x_{i} , x_{i+2} , \ldots , x_M)$.

\item The path $p_2$ consists of $x_{i}-x_{i+1}$ level moves at the temperature $i$ from  \\ $(x_0 ,\ldots x_{i-1}, x_{i} , x_{i} ,  x_{i+2} , \ldots , x_M)$ to $x'$.
\end{itemize}
If, on the other hand, $x_{i+1} > x_i \ge 0$, the canonical paths are defined as the concatenation $p_1 \circ p_2$ where
\begin{itemize}
\item The path $p_1$ consists of $x_{i+1} - x_i$ level moves at the temperature $i$ from $x$ to \\ $(x_0,\ldots,x_{i-1},x_{i+1},x_{i+1},x_{i+2},\ldots,x_M)$.
\item The path $p_2$ consists of $x_{i+1} - x_i$ level moves at the temperature $i+1$ from \\ $(x_0,\ldots,x_{i-1},x_{i+1},x_{i+1},x_{i+2},\ldots,x_M)$ to $x'$.
\end{itemize}
The idea behind the definition of the canonical paths is to use the higher probability state (either $x_i$ or $x_{i+1}$) to ensure the edges along the path will always have sufficiently high weight.

We can bound the factor $A$ in Theorem \ref{thm:refined-compthm} as follows. Firstly note that the paths are at most polynomial in length and the number of transitions of $P_t$ which utilize any transition of $\hat P_t$ are at most polynomial in number since there are only a polynomial number of possible states $x_i$ and $x_{i+1}$. Hence, it is enough to show that for any transition $(z,w)$ of $\hat P_t$ and any transition $(x,x')$ of $P_t$ such that $(x,x') \in \Gamma(z,w)$, the quantity
\[
\frac{\pi_t(x)P_t(x,x')}{\pi_t(z) \hat P_t(z,w)}
\]
is at most a polynomial, where $\pi_t$ denotes the stationary measure for $P_t$ as well as $\hat P_t$. This can be verified in a straightforward way by checking the possible cases. For example, assume that the transition $(z,w)$ changes the $i+1$st coordinate of $z$. Then, either it is on a path $p_1$ when $x_i \ge x_{i+1}$ or it is on a path $p_2$ and $x_{i+1} > x_i$.  
Consider the first case that $x_i \ge x_{i+1}$. In this case,  and $z_{i} = w_i= x_{i}$ and we obtain
\begin{align*}
\frac{\pi_t(x)P_t(x,x')}{\pi_t(z) \hat P_t(z,w)} =  \frac{\pi_i(x_i) \pi_{i+1}(x_{i+1})\min\left(1, \frac{\pi_{i+1}(x_i)
\pi_i(x_{i+1})}{\pi_i(x_i)\pi_{i+1}
(x_{i+1})}\right)}{ \pi_i(x_i)\pi_{i+1}(z_{i+1})\min\left(1,\frac{\pi_{i+1}(w_{i+1})}{\pi_{i+1}(z_{i+1})} \right)} = \frac{\min(C^{(i+1)x_{i+1}/M},C^{(x_i+ix_{i+1})/M})}{\min(C^{(i+1)z_{i+1}/M}, C^{(i+1)w_{i+1}/M})}
\end{align*}
The last expression can be simplified using the fact that by construction of the path, $z_{i+1} < w_{i+1}$. Finally, by the construction of the path and the fact that the level moves preserve the trace, we also know that $z_{i+1} \ge x_{i+1}$. Hence, we obtain
\begin{align*}
\frac{\pi_t(x)P_t(x,x')}{\pi_t(z) \hat P_t(z,w)} = C^{(i+1)(x_{i+1} - z_{i+1})/M} \le 1.
\end{align*}
A similar calculation made for the case that $(z,w)$ is an edge on the path $p_2$ so that $x_{i+1} > x_i $, $x_i \le w_{i+1} <  z_{i+1} $ and $w_i=z_i = x_{i+1}$ shows that

\begin{align*}
\frac{\pi_t(x)P_t(x,x')}{\pi_t(z) \hat P_t(z,w)} = C^{(i+1)(x_i - w_{i+1})/M} \le 1.
\end{align*}

Combining the bound on $A$ with the polynomial mixing of $\hat P_t$ and applying Theorem \ref{thm:refined-compthm}, we conclude that the restriction chains $P_t$ mix in polynomial time. \qedhere

\end{proof}

\vskip.1in
\noindent \underbar{\bf Mixing time of the projection:}
\vskip.05in
\noindent
The stationary probabilities of the projection chain are given by
\[\overline{\pi}\left({t}\right) =
\displaystyle\sum_{x : Tr({x})={t}} \pi_{sw}({x}).\]
It can be verified that $\overline\pi$ is also a product
measure over the temperatures. Let $\overline\pi_i$ denote the two point distribution at each temperature such
that $\overline\pi$ is the product of the $\{\overline\pi_i\}$.

To show the projection $\overline P$ mixes in polynomial time, we will compare it to the following simpler chain $\widetilde P$ on the $M+1$ dimensional hypercube and with the same stationary distribution. In $\widetilde P$,at each step we are allowed
to transpose two neighboring bits, or we can flip just the first bit.
Each of these
moves is performed with the appropriate Metropolis probability.   This captures the idea that for
the true projection chain $\overline P$, swap moves (corresponding to transpositions of bits) always have constant
probability, and that at the highest temperature there is high probability of
changing sign. Of course there is in addition the chance of flipping
the bit at each lower temperature, but this seems to be a smaller effect.

More formally, at each step in $\widetilde P$, we pick $i \in_u \{0,\ldots,M\}$ and update
the $i^{th}$ component $t_i$
by choosing $t_i'$ with probability $\overline{\pi}_i(t_i')$, i.e., exactly
according to the appropriate
stationary distribution.  In other words,
the $i^{th}$ component is at stationarity as soon as it is chosen.
Using the coupon collector's theorem, we have

\begin{lemma} \label{lema:comparison-chain} The chain $\widetilde{P}$
  on $\{0,1\}^{M+1}$
mixes in time $O\left(M \log(M+\varepsilon^{-1})\right)$ and
$Gap(\widetilde{P})^{-1}= O(M\log M)$.
\end{lemma}

We are now in a position to prove the following lemma.
\begin{lemma}
\label{rapid}
The projection $\overline{P}$ of the swapping Markov chain is
polynomially mixing
on $\{ 0, 1\}^{M+1}$.
\end{lemma}
\begin{proof}
To apply the comparison theorem, we translate
transitions in the chain $\widetilde{P}$, whose mixing time we know,
into a canonical path consisting of moves in the chain
$\overline{P}$. Let
$(t, t')$ be a single transition in $\widetilde{P}$ from
${t}=(t_0,\ldots,t_i,\ldots,t_M)$
to ${t'}=(t_0,\ldots,1-t_i,\ldots,t_M)$ that flips the $i^{th}$ bit.
The canonical path from ${t}$
to ${t'}$ is the concatenation of three paths $p_1 \circ p_2 \circ p_3$.
In terms of tempering, $p_1$ is a heating phase and $p_3$
is a cooling phase.
\begin{itemize}
\item The path $p_1$ consists of $i$ swap moves
from ${t}$ to $(t_i, t_0,\ldots,t_{i-1}, t_{i+1},\ldots,t_M)$;
\item The path $p_2$ consists of one step that flips the bit corresponding
  to the highest temperature to move to
$(1-t_i, t_0,\ldots,t_M)$;
\item The path $p_3$ consists of $i$ swaps until we reach
${t'}=(t_0,\ldots,1-t_i,\ldots,t_M)$.
\end{itemize}

To bound $A$ in Theorem \ref{thm:refined-compthm}, we will establish that
\begin{equation}\label{eqn1}
\overline{\pi}({z})\ \overline{P}({z}, {z'}) \geq
\frac{1}{2} \overline{\pi}({t})\ \widetilde{P}({t},{t'}),
\end{equation}
for any transition $({z}, {z'})$
in the canonical path.
Second, we need to ensure that the number of paths using the
transition $(z, z')$,
$\Gamma_{z,z'}$, is at most a polynomial.
These two conditions are sufficient to give a polynomial bound on the parameter
$A$ in the comparison theorem.
For any $(z,z')$ we have $|\Gamma(z,z')|\leq M^2$, so it remains to
establish the condition in Equation~\ref{eqn1}.

\vskip0.1in
\noindent \underbar{Case 1:} Transitions along $p_1$.

Let ${z}=(t_0,\ldots,t_{j-1}, t_i, t_{j},\ldots,t_{i-1}, t_{i+1},\ldots,t_M)$
and ${z'}=(t_0,\ldots,t_i, t_{j-1}, \ldots, t_{i-1}, t_{i+1}, \ldots, t_M)$.
\begin{eqnarray}\label{mineqn}
\overline{\pi}({z})\overline{P}({z}, {z'})
& = & \frac{\overline{\pi}({z})}{2(M+1)}
\min\left( 1, \frac{\overline{\pi}({z'})}{\overline{\pi}({z})}\right) \\
& = & \frac{1}{2(M+1)} \min\left(\overline{\pi}({z}),
\overline{\pi}({z'})\right). \nonumber
\end{eqnarray}
First we consider $\overline{\pi}({z})$.
\begin{eqnarray*}
\overline{\pi}({z}) = \prod_{\ell=0}^M
\  \displaystyle\sum_{Tr(x)_{\ell} =  z_{\ell}} \pi_{\ell}(x)
\triangleq \prod_{\ell=0}^M \pi_{\ell}(z_{\ell}).
\end{eqnarray*}
Assume, without loss of generality, that $N \leq N'.$
 Then we have
\begin{eqnarray*}
\overline{\pi}({t})\widetilde{P}({t}, {t'})
& = & \frac{\overline{\pi}(t)}{M+1}
\overline{\pi}_i(1-t_i)\\
& \leq & \frac{\min\left(\overline{\pi}_i(t_i),
    \overline{\pi}_i(1-t_i)\right) }{M+1} \prod_{j \neq i}
\overline{\pi}_j(t_j)\\
& = & \frac{\overline{\pi}({t^*})}{M+1},
\end{eqnarray*}
where ${t^*} = (t_0,\ldots,t_{i-1},0,t_{i+1},\ldots,t_M)$.
We want to show that $\overline{\pi}({t^*}) \leq \overline{\pi}({z}).$
It is useful to partition $t^*$ into blocks of bits $t_{\ell}$ that equal~1,
separated by one or more zeros.
Let $k < i$ be the
largest value such that $t_k=0$. It can be verified from the
definition of the distribution that
\begin{eqnarray*}
\pi_i(1)\pi_{k+1}(0)  \geq \pi_i(0)\pi_{k+1}(1)
\end{eqnarray*}
From this fact, it follows that
$$\prod_{\ell=k+1}^i \pi_{\ell}(z_{\ell}) \geq \prod_{\ell=k+1}^i
\pi_{\ell}(t_{\ell}^*).$$
Similarly, considering the next block of $t^*$ (i.e., the next set of
bits such that
$t_{\ell}=1$)
until the first index $k'$ such that $t_{k'}=0$,
$$\prod_{\ell=k'+1}^k \pi_{\ell}(z_{\ell}) \geq \prod_{\ell=k'+1}^k
\pi_{\ell}(t_{\ell}^*).$$
Continuing in this way we find
$$\prod_{\ell=j}^i \pi_{\ell}(z_{\ell}) \geq \prod_{\ell=j}^i
\pi_{\ell}(t_{\ell}^*),$$
and thus
$$\overline{\pi}({z}) \geq \overline{\pi}({t^*}).$$
Likewise, by taking one more term, we find that $\overline{\pi}({z'})
\geq \overline{\pi}({t^*}).$
Together with equation~\ref{mineqn} this implies
\begin{eqnarray*}
\overline{\pi}({z})\ \overline{P}({z},{z'}) \geq
\frac{1}{2} \overline{\pi}({t})\ \widetilde{P}({t},{t'}).
\end{eqnarray*}

\vskip.1in
\noindent \underbar{Case 2:} The transition along $p_2$.
Consider the transition from ${z}=(t_i,t_0,\ldots,t_{i-1},t_{i+1},\ldots,t_M)$
to ${z'}=(1-t_i, t_0,\ldots,t_M)$ that flips the first bit of $z$.
Repeating the argument from Case 1, it follows that
$$\min\left(\overline{\pi}({z}), \overline{\pi}({z'})\right)
 \geq \overline{\pi}({t^*}).$$
Therefore, again we find equation~\ref{eqn1} is satisfied.
\vskip.1in
\noindent \underbar{Case 3:} Transitions along $p_3$.
This is similar to Case 1.

\vskip.1in
In all three cases, we find that if $({z}, {z'})$ is one
step on the canonical
path from ${t}$ to ${t'}$, equation~\ref{eqn1} is satisfied.
Therefore, it follows that
$$
A=\max_{(z,z')\in E(\overline{P})}
\left\{
\frac{\displaystyle\sum_{\Gamma(z,z')} |\gamma_{t,t'}|\overline{\pi}(t)
\widetilde{P}(t,t')}{\overline{\pi}({z}) \overline{P}({z},{z'})}
\right\} \  = O(M^3).$$
Hence, by Lemma \ref{lema:comparison-chain}, applying Theorem
\ref{thm:refined-compthm}, $Gap(\overline{P})  = \Omega(M^{-4}\ln M).$

\end{proof}

This establishes all the results necessary to apply the
decomposition theorem Theorem \ref{dcmp}, completing the proof of
Theorem \ref{expfast} by Theorem \ref{thm:spec_gap_thm}.

\section{Mean-Field Models}\label{newswap}

The analysis from the last section suggests how to design
distributions for swapping and tempering in cases where the
mixing time is not known or known to be exponentially large. We
consider examples of mean-field
models to illustrate the ideas.  While the examples below are
very specific to mean-field models, our results indicate that there
are more robust methods for designing tempering and swapping
algorithms.


\vskip.1in
\noindent \underbar{\bf Mean-Field Ising Model with an External
  Field: }
An important special case of the $q$-state Potts model with an
external field is the mean-field Ising model in the presence
of an external field.  This model is defined by parameters $q=2$,
$\beta  \ge 0$, the inverse temperature, and $h$, the external
magnetic field.
The Gibbs distribution over configurations $x\in \Omega = \{+1,-1\}^V$ is
\[\pi(x)=\pi_{\beta, h}(x) =  \frac{\exp\left(\beta \left( \sum_{i<j
      }\delta_{x_i,x_j}+ h\sum_{i=1}^2\delta_{x_i,1}
      \right)\right)}{Z(\beta, h)},\]
where  $Z(\beta, h)$ is the normalizing constant.
We will show that with a modified set of
        appropriately ``dampened" distributions, swapping
        can be used to sample configurations in this case.

\subsection{Entropy Dampening Distributions}
\label{sec:ent-damp}
Traditionally, a convenient choice for the swapping and tempering
distributions are the tempered distributions given by the Gibbs
distributions $\pi_{i}$ for a chosen sequence of inverse temperatures.

The idea for the new distributions we define stems from the
observation that this may be a poor choice of interpolants
because they preserve the first-order phase transition, as we will
show in the next section is the case for the Potts model.  We can do
much better by exploring a wider class of interpolating distributions.
To see the flexibility we have in defining the set of distributions, define
$$\rho_i(x)= \frac{\pi_i(x) f_i(x)}{Z_i},$$
where $Z_i=\sum_{x\in\Omega}\pi_i(x) f_i(x)$ is another normalizing constant.
When $f_i(x)$ is taken to be the constant function, then we obtain the usual
tempered distributions. Recall that in the mean-field model with $q$ spins, we let $\sigma = (\sigma_1,\ldots,\sigma_q)$ denote the numbers of vertices with spins $1,\ldots,q$.
Define
$$f_i(x) = {n \choose \sigma_1, \ldots, \sigma_q}^{\frac{i-M}{M}}.$$ The {\em
  Flat-swap} algorithm or chain is then defined to be the swapping algorithm using
  the distributions $\rho_0, \ldots,\rho_M$ as defined above. We
  define the {\em Flat-tempering} algorithm analogously using the same
  set of distributions.

\subsection{Polynomial Mixing of Flat-Swap}
For a configuration $x \in \Omega$, recall that we let $\sigma_i$ be the number of vertices colored $i$ and let $\sigma = \sigma(x)=(\sigma_1,\ldots,\sigma_q)$, where $\sum_q
\sigma_q=n$. Define $\Omega_{\sigma} \subset \Omega$ to be the set of configurations with $\sigma_i$
vertices assigned
color $i$.
The {\it total spins distribution} is the discrete distribution on the set of possible $\sigma$, 
$$S_{\sigma}=\pi(\Omega_{\sigma}) = \sum_{x\in\Omega_\sigma} \pi(x).$$
For the Ising model, we set $$f_i(x)={n\choose{k}}^{\frac{i-M}{M}},$$
if in the configuration $x$, $k$ vertices are assigned $+1$ and $n-k$ are assigned $-1$.  Let $\beta_i= {\beta}\cdot \frac{i}{M}$.   Note that
$f_i(x)$ is easy to compute given $x$.
 A simple calculation shows that
\begin{align}
\rho_i(\Omega_{(k,n-k)}) \ = \ {n\choose{k}} \rho_{i}(x)
\ = \  \frac{1}{Z_i}\left(\rho_{M}(\Omega_{(k,n-k)})\right)^{\frac{i}{M}}.
\label{eq:dampened-dist}
\end{align}
The function $f_i(x)$ effectively dampens the entropy (multinomial)
just as the change in
temperature dampens the energy term coming from the Hamiltonian.
Thus, all the total spins distributions have
the same relative shape, but get flatter as $i$ is decreased.
This no longer preserves the cut in the state space of the
distributions for the usual swap algorithm.
It is this property that makes this choice of distributions useful.

\begin{theorem}\label{flatthm}
The Flat-swap algorithm mixes polynomially for every inverse temperature
$\beta>0$ and any external field $h$ for the Ising model.

\end{theorem}

We follow the strategy set forth in the proof of Theorem \ref{expfast}, using
decomposition and comparison in a similar manner.
The total spins distribution for the Ising model
is known to be bimodal above the critical temperature (at and below
the critical temperature there is a unique maximum), even in the
presence of an external field.
With our choice of distributions $\rho_i$, it now follows that all
$M+1$ total spins distributions are bimodal as well.
Moreover, the minima of the distributions occur
at the same value of $i\in [n]$ for all $M+1$
distributions.  Let $t_{\min}$ be the value of $i$ which is the
minimum. Let $\sigma^0_{\max}$ and $\sigma^1_{\max}$ denote the
equivalence
classes of configurations which maximise the total spins distributions
for $i=M$ on either side of $t_{\min}$.

Let $\Omega_{sw}=\Omega^{M+1}$ be the state space of the chain.
Define the trace Tr$({x})={t} \in \{0,1\}^{M+1}$,
where $t_i=0$ if the number of $+1$s in $x_i$ is less than $t_{\min}$ and
let $t_i=1$ if the number of $+1$s in $x_i$ is at least $t_{\min}$.
As in the case of the exponential distribution, we partition
$\Omega_{sw}$ according to the trace of the state $x$. Let
$\overline{P}$ be the projection Markov chain for this partition and let $P_t$ denote the corresponding restriction chains.

We begin by showing that the projection chain mixes in polynomial time. The idea of the proof that the restrictions mix polynomially is analogous to the arguments in Lemma \ref{lem:poly-mixing-fixed-trace}, although the details are slightly different for the Ising model, and we make use of results of \cite{MZ}.

Let
$\overline{\rho}$ be the stationary distribution, which is the
product of the distributions $\overline{\rho}_i$ each of which is a
distribution on the two point space $\{0,1\}$. 
Without loss of generality let $0$ be the coordinate for which the the
mode is lower at the temperature $M$.

\begin{lemma}
At every temperature $i$, $\overline{\rho}_i(\sigma^0_{\max}) \leq
\overline{\rho}_i(\sigma^1_{\max})$.
\end{lemma}
\begin{proof} By \eqref{eq:dampened-dist}, 
$$\overline{\rho}_i(\sigma^0_{\max})  = \frac{1}{Z_i}
(\rho_M(\sigma^0_{\max}))^{\frac{i}{M}} \leq \frac{1}{Z_i}
(\rho_M(\sigma^1_{\max}))^{\frac{i}{M}}  = \overline{\rho}_i(\sigma^1_{\max}).$$
 \end{proof}

\begin{lemma}\label{lem:63} At every temperature $i$, for $z \in \{0,1\}$,
  $\overline{\rho}_i(z)$ is within a factor of $O(n)$ of
  $\rho_i(\sigma_{\max}^{z})$.
\end{lemma}
\begin{proof} Clearly, $\overline \rho_i(z) \ge \rho_i(\sigma_{\max}^{z})$. On the other hand, there are $O(n)$ equivalence classes of configurations of which
  $\sigma_{\max}^{z}$ is one and has the largest relative
  weight. Hence, $\overline \rho_i(z) \le O(n) \rho_i(\sigma_{\max}^{z})$.
 \end{proof}

\begin{corollary}\label{cor:01}
For every pair of temperatures $i>j$,
$$ \overline{\rho}_i(0) \overline{\rho}_j(1) \leq O(n^2) \overline{\rho}_i(1) \overline{\rho}_j(0) $$
\end{corollary}
\begin{proof}By Lemma \ref{lem:63},
\begin{eqnarray*}
\frac{\overline{\rho}_i(0)}{ \overline{\rho}_j(0)} &  \leq & O(n)
\frac{\rho_i(\sigma_{\max}^{0})}{\rho_j(\sigma_{\max}^{0})} \\
& =  & O(n) \frac{Z_j}{Z_i}
(\rho_M(\sigma_{\max}^{0}))^{\frac{i-j}{M}} \\
& \leq  & O(n) \frac{Z_j}{Z_i}
(\rho_M(\sigma_{\max}^{1}))^{\frac{i-j}{M}} \\
& \leq  & O(n^2) \frac{\overline{\rho}_i(1)}{ \overline{\rho}_j(1)}
\end{eqnarray*}
 \end{proof}

\begin{theorem} The projection Markov chain $\overline{P}$ is
  polynomially mixing on $\{0,1\}^{M+1}$.
\end{theorem}

\begin{proof} We appeal to the comparison theorem. Let $\widetilde{P}$
  be the Markov chain on $\{0,1\}^{M+1}$ whose transitions choose a
  random index $i$ and update $t_i$ to $t_i' \in \{0,1\}$ with probability
  proportional to $\overline{\rho}_i(t_i')$. It is clear that $\widetilde P$ mixes in polynomial time since whenever the temperature $i$ is chosen, the corresponding coordinate is at stationarity after the update. Since the temperatures are chosen uniformly at random among the $M+1$ possible temperatures, standard coupon-collector arguments imply that the mixing time of $\widetilde P$ is $O(M\log M)$.
  
   Let $t=(t_0,
  \ldots,t_i,\ldots,t_{M})$ and $t'=(t_0,\ldots,t'_i,\ldots,t_{M})$ be
  two states such that $\widetilde{P}(t,t')>0$. We define a path
  between them using transitions of $\overline{P}$. Assume that $t_i = 0$. In the other case, define the path to be the reverse.
Denote the path by $t = z^0, \ldots,z^\ell = t'$.
From $z^m$, define $z^{m+1}$ as follows:

\begin{itemize}
\item Let $j$ be the largest index such that $z^m_j = 0$
and $t'_j = 1$.
\item If there is a largest index $k< j$, such that
  $z^m_k = 1$, then $z^{m+1}$ is obtained by swapping the
  bits in positions $k$ and $k+1$.
\item If there is no such $k$, $z^{m+1}$ is obtained by flipping the bit
  $z^m_0$ from 0 to 1.
\end{itemize}

The idea of the path is to flip the bit at the index $j$ from 0 to 1
by moving a 1 up from the first available position, performing a
series of swaps through a
block of 0's. At the end of the series of swaps,
the difference at $j$ has been removed, and the difference is now at
the starting point of the swaps.
As the bit 1 moves up, there can be at most 2 differences
due to it: one at its current position and the other at the position
where the series of swaps began. Hence there are at most 4 indices
where there could be a difference between $t$ and
$z^m$ (the other two being the index $j$ and possibly the index
$i$). Moreover, by the construction, the indices must be such that for
the highest, say $i_1$, $t_{i_1} = 0$ while $z^m_{i_1}  =1$ and the
difference then alternates. Thus, we have that for any configuration
$z^m$ along the path, by Corollary \ref{cor:01},

$$\frac{\overline{\rho}(t)}{\overline{\rho}(z^m)} \leq O(n^4).$$

Suppose that we fix a transtion $z,z'$. The number of pairs
$t,t'$ such that the path between them passes through $z,z'$ can be
bounded by $M^4$ since we must only specify the positions at which
$z$ differs from $t$ and possibly the index $i$  to be able to
reconstruct both $t$ and $t'$.
We can now bound the factor $A$ in the comparison theorem Theorem \ref{thm:refined-compthm} as follows

\begin{eqnarray*}
A & = & \max_{(z,z')\in E(\overline{P})}
\left\{
\frac{\displaystyle\sum_{\Gamma(z,z')} |\gamma_{t,t'}|\overline{\rho}(t)
\widetilde{P}(t,t')}{\overline{\rho}({z}) \overline{P}({z},{z'})}
\right\} \\
& \leq & \max_{(z,z')\in E(\overline{P})}
\left\{
\frac{ 2 \displaystyle\sum_{\Gamma(z,z')}
  |\gamma_{t,t'}|
  \overline{\rho}(t)}{\min(\overline{\rho}({z}),\overline{\rho}(z'))}
\right\} \\
& \leq & O(n^4M^5).
\end{eqnarray*}
The claim follows by applying Theorem \ref{thm:refined-compthm} and using the fact that $\widetilde P$ mixes in polynomial time. \qedhere

\end{proof}

Recall that $P_t$ is the restriction chain with a fixed trace $t \in \{0,1\}^{M+1}$ and let $\hat P_t$ denote the restriction chain where the swap moves are suppressed. Then, $\hat P_t$ consists of independent chains on each of the $M+1$ distributions. It was shown in \cite{MZ} that in the case of zero external field, each of the $\hat P_t$ are rapidly mixing by using Lemma \ref{lem:product}, the decomposition theorem and comparison of the restrictions with a simple exclusion process on the complete graph. The same analysis holds in our case, the main difference being the non-zero external field. Since the trace is fixed for each of the chains however, the only fact that must be checked is the analog of Lemma 14 of \cite{MZ} which says that the distributions at each temperature are unimodal on either side of $t_{\min}$. We omit the calculations as they are straightforward to check.
Thus, we have that $\hat P_t$ is polynomially mixing for each $t \in \{0,1\}^{M+1}$.

\begin{lemma}
For each trace $t \in \{0,1\}^{M+1}$, the restriction chain $P_t$ mixes in polynomial time.
\end{lemma}
\begin{proof}
As in Lemma \ref{lem:poly-mixing-fixed-trace}, the strategy is to use the comparison theorem and compare to the chain $\hat P_t$. However, the state space at each temperature is no longer an interval and thus there are some additional steps. We will show that the 2-step chain $Q_t = P_t^2$ is polynomially mixing. This implies polynomial mixing for $P_t$. To show that $Q_t$ mixes polynomially, we use decomposition. Let $Q_{t,\sigma^0,\cdots,\sigma^M}$ denote the restriction of the chain where at each temperature, not only is the trace fixed to $t$, but the spin configuration at temperature $i$ is in the set $\Omega_{\sigma^i}$. The projection chain $\overline Q_t$ thus moves on the sets $\Omega_{\sigma^i}$ at each temperature $i$.

Note that the restriction chains $Q_{t,\sigma^0,\cdots,\sigma^M}$ are exactly the same as the restriction chains for $(\hat P_t)^2$ when the restrictions fix the spin configuration at each temperature. The rapid mixing of this chain follows by the arguments in \cite[Section 7]{MZ}.

Thus, we are reduced to showing that $\overline Q_t$ is polynomially mixing and we do this by comparing to $\overline {\widehat Q_t}$, the projection on total spins of the two step chain when swap moves are suppressed. This can be done along the lines of the comparison proof in Lemma \ref{lem:poly-mixing-fixed-trace}. 

Let $(x,x')$ be a transition of $\overline Q_t$ with $x = (\sigma_0,\ldots,\sigma_M)$ and $x' = (\sigma'_0,\ldots,\sigma'_M)$. If $(x,x')$ is a level move which updates the state at a fixed temperature, let the corresponding canonical path in $\overline {\widehat Q_t}$ be the edge $(x,x')$ itself. 
On the other hand, suppose $(x,x')$ is a temperature move so that for some $1 \le i < M$, $\sigma_j = \sigma_j'$ for all $j \not\in \{i,i+1\}$, $\sigma_i = \sigma'_{i+1}$ and $\sigma_{i+1} = \sigma'_i$. In this case, we divide the construction of the path into several cases based on the values of $\sigma_i$ and $\sigma_{i+1}$. Note that since the trace remains fixed, it must be the case that $t_i=t_{i+1}$. Without loss of generality, we assume that $t_1 = 1$, since the calculation in the other case is exactly the same.
\begin{enumerate}[1)]
\item In the first case, the states $\sigma_i,\sigma_{i+1} \ge \sigma^1_{\max}$ or $\sigma_i,\sigma_{i+1} < \sigma^1_{\max}$, that is, they are on the same side of the state $\sigma^1_{\max}$. In these cases, the construction of the canonical path is analogous to the construction in Lemma \ref{lem:poly-mixing-fixed-trace}.

\item The the second case, $\sigma_i< \sigma^1_{\max}$ and $\sigma_{i+1} \ge \sigma^1_{\max}$ or $\sigma_i \ge  \sigma^1_{\max}$ and $\sigma_{i+1} < \sigma^1_{\max}$. The path consists of the concatenation of four paths. First, we move with level moves at the temperature $i$ from $\sigma_i$ to $\sigma^1_{\max}$. Next, we move at temperature $i+1$ with level moves from $\sigma_{i+1}$ to $\sigma^1_{\max}$. So far, the stationary weight of states along the path are non-decreasing. The next part consists of two non-increasing paths. First, we move at temperature $i$ from $\sigma^1_{\max} $ to $\sigma_{i+1}$. last, we move at the temperature $i+1$ from $\sigma^1_{\max} $ to $\sigma_{i}$.
\end{enumerate}

It can be verified that since the stationary measure along the paths is unimodal, the length of any path is polynomial and there are at most polynomially many paths using any transition, the comparison constant $A$ can be bounded above by a polynomial.

\end{proof}

Since the projection chain $\overline P$ and each of the
restrictions $P_t$ mixes in time that is polynomial in $n$ and $M$, Theorem \ref{flatthm} follows.


\section{Torpid mixing of Simulated Tempering}
\label{torpid}
In this section we show that for the mean-field 3-state ferromagnetic Potts model, there is a critical temperature so that for any  distribution parametrized by temperature, the mixing time of the tempering and swapping algorithms
is exponential.\\

\noindent{\bf Theorem \ref{thm:slow-potts}}
{\em Let $\beta_c = \frac{4\log2}{n}$. There is a constant $c_1 > 0$ such that for any set of inverse temperatures $\beta_c = \beta_M \ge \cdots \ge \beta_0 \ge 0$ such that $M = n^{O(1)}$, the
tempering and swapping chains with the distributions $\pi_{\beta_i}$ for the $3$-state mean-field ferromagnetic Potts model have mixing time
$\tau(\varepsilon) \geq
e^{c_1n}\ln(1/\varepsilon)$.}  \\

We prove the lower bound on the mixing time of the
tempering chain by bounding the conductance. The slow mixing on the swapping algorithm then follows by Zheng's result \cite{Z} showing that polynomial mixing of the swapping chain implies polynomial mixing of the simulated tempering chain with the same distributions.
The \textit{conductance} is an isoperimteric quantity
related to the spectral gap through Cheeger's inequality, a version of which was shown independently by Jerrum and Sinclair \cite{JS-conductance} and Lawler and Sokal \cite{LawSok88}. It often gives an easier method for bounding the mixing
time than directly bounding the spectral gap.
For $S \subset \Omega$, let $$\Phi_S = \frac{F_S}{C_S}
= \frac{\displaystyle\sum_{x\in S, y\notin S}
\pi(x)P(x,y)}{\pi(S)}.$$ Then, the conductance
is given by $$\Phi=\min_{S: \pi(S)\leq 1/2} \Phi_S$$ and it
bounds the mixing time both from above and below. Cheeger's inequality implies the following bounds on the mixing time. Let $\pi_{\min} = \displaystyle\min_{x \in \Omega} \pi(x)$.
\begin{theorem}\label{thm:mixing-conductance}
For any reversible Markov chain with conductance $\Phi$
$$ \frac{1-2\Phi}{2\Phi} \log\left( \frac{1}{2\varepsilon}\right) \leq \tau(\varepsilon) \leq
\frac{1}{\Phi^2}\left(\log\left( \frac{1}{2\varepsilon}\right) + \frac 12 \log \left( \frac{1-\pi_{\min}}{\pi_{\min}} \right)\right).$$
\end{theorem}

The state space of the tempering chain
is $\Omega \times [M+1]$ where $\Omega$ consists of spin configurations on the complete graph with three types of spins. To show torpid mixing, it is enough to
exhibit a cut in the state space whose conductance is small.
For convenience, let us call the 3 spins red, blue and green. The cut we construct depends only on the {\em number} of red, blue and
green vertices in the configuration. Hence, for the purpose of
defining the cut, it is convenient to divide the state space of
configurations $\Omega$ into equivalence classes of colorings according to the
number of vertices of each color. Furthermore, the cut we define will
induce the same cut on $\Omega$ at each temperature.

It is convenient for the exposition to make the
following reparametrization using the fact that
for the mean-field Potts, the underlying graph is complete.
Let $\overline H(x)=\sigma_1^2+\sigma_2^2+\sigma_3^2$, let $\overline \beta =
\beta/2$ and let $\overline{Z}(\beta)$ denote the corresponding partition function. It can be verified that the Gibbs distribution at inverse temperature $\beta$ can be written as
\[
\pi_\beta(x) = \frac{e^{\overline \beta \overline H(x)}}{\overline
  Z(\beta)}.
\]
To define the cut, we partition $\Omega$ into sets
$\Omega_{\sigma}$, where
$\sigma=(\sigma_1, \sigma_2, \sigma_3)$ is partition of $n$ and $\Omega_{\sigma}$ contains all colorings with $\sigma_1,
\sigma_2$ and $\sigma_3$ vertices colored red, green and blue,
respectively. It is helpful to think of the $\sigma$ as points on a simplex.
The set $\Omega_\sigma$ corresponds to
${n \choose \sigma_1,\sigma_2,\sigma_3}$ different configurations in
$\Omega$ and hence we write
\begin{eqnarray}
\label{eq:gibbs}
\pi_{\beta_i}(\Omega_\sigma) =  {n \choose
\sigma_1,\sigma_2,\sigma_3}
\frac{e^{\overline \beta_i(\sigma_1^2+\sigma_2^2+\sigma_3^2)} }{\overline Z(\beta_i)} \
\end{eqnarray}

The idea for defining the cut with small conductance comes from the
following properties of the stationary distribution conditioned on the sets
$\Omega_{\sigma}$.
There is a critical temperature
$\beta_c$ where the Gibbs distribution exhibits the coexistence of
two modes. There is a ``disordered''
mode in the distribution at
$\left(\frac{n}{3},\frac{n}{3},\frac{n}{3}\right)$; this mode
is present because though these configurations have small energy,
the number of configurations (given by the multinomial term in
Equation (\ref{eq:gibbs})) is large.
At $\beta_c$, there are also ``ordered'' modes at
$\left(\frac{2n}{3},\frac{n}{6},\frac{n}{6}\right)$,
$\left(\frac{n}{6},\frac{2n}{3},\frac{n}{6}\right)$,
$\left(\frac{n}{6},\frac{n}{6},\frac{2n}{3}\right)$. These modes are
present because configurations with a
predominant number vertices having the same color (red, or green or
blue) are favored in the
Gibbs distribution, though there are not as many of these
configurations. The ordered and disordered modes are separated by a
region whose
density is exponentially smaller than both
the modes, where neither the multinomial nor the energy term
dominates. As the inverse temperature is decreased below $\beta_c$, the size
of the disordered mode grows while the sizes of the ordered
modes decrease. However,
the region of exponentially small density remains small at every
temperature. The cut in the
state space of the simulated tempering chain at $\beta_c$ is to take a
region surrounding the ordered mode at each temperature. The
conductance of this cut, up to a polynomial (in $M$) is bounded by the
conductance at the critical temperature where the modes coexist. This
is because in the stationary distribution, the chance of being at each
temperature is equally likely. In contrast, for the Ising model, there
is no temperature at which the ordered and disordered modes coexist.
We first present a straightforward upper bound on the conductance of
the tempering chain at $\beta_c$.
\begin{theorem}\label{thm:small-conductance}
Let $\beta_c = \frac{4\log 2}{n}$. There is a constant $c_4 > 0$ such that
the conductance $\Phi$ of the simulated tempering chain with distributions $\{ \pi_{\beta_i} \}$ for any $\beta_c = \beta_M \ge \cdots \ge \beta_0 \ge 0$ for
  the 3-state mean-field ferromagnetic Potts model is at
  most  $e^{-c_4n}$.
\end{theorem}
The lower bound on
mixing time by the inverse of conductance in Theorem \ref{thm:mixing-conductance} implies Theorem
\ref{thm:slow-potts}. In the next section, we will refine this bound in
order to compare it to the upper bound on the mixing time of the
Metropolis chain at a fixed temperature to show Theorem
\ref{thm:temp-vs-metropolis}.

Let $A \subset \Omega$ be the set of configurations such that
$\sigma_1,\sigma_2,\sigma_3 \leq n/2$. Let $P_{st}$ denote the transition matrix of the simulated tempering chain. Let $S =\{(x,i) \
| \ x \in A, \ \beta_0 \le \beta_i \le \beta_c \}$.
Let $$B = \{x \in A \ |
\ \exists \ x' \in \Omega \setminus A, \ P_{st}((x,i),(x',i)) > 0 \ \forall \ 0 \le i \le M\}$$ be the boundary of
$A$ (the set of configurations with at least one of
$\sigma_1,\sigma_2$ or $\sigma_3$ equal to $n/2$).
Our aim is to show that the conductance $\Phi_S$ of the set $S$ is
exponentially small. Note that it is not true that $\pi(S) \leq 1/2$ and hence
this does not immediately imply a bound on $\Phi$. Instead, we will
show that the coexistence of the ordered and disordered phases implies
that $\Phi \leq n^{O(1)} \Phi_S$. We start by bounding $\Phi_S$.

\begin{align}
\nonumber \Phi_S  = \frac{F_S}{C_S}     & = 
 \frac{\displaystyle\sum_{i \in I}\displaystyle\sum_{x \in B} 
\pi_{\beta_i}(x)
\displaystyle\sum_{x' \in
\overline{A}}  P_{st}((x,i),
(x',i))}{\displaystyle\sum_{i \in I}\displaystyle\sum_{x \in A}
\pi_{\beta_i}(x)} \\
& \leq   \frac{\displaystyle\sum_{i \in
 I}\displaystyle\sum_{x \in B} 
\pi_{\beta_i}(x)}{\displaystyle\sum_{i \in I}\displaystyle\sum_{x \in A}
\pi_{\beta_i}(x)} \label{eq:cond1}
\end{align}

The last expression above is the ratio of the sum over temperatures of the
stationary probabilities of configurations in the set $B$ (the
boundary of the set $A$) to the sum over temperatures of the
stationary probabilities of the configurations in the set $A$. In
order to bound this quantity, we will need several technical lemmas
which we state in the course of the proof but prove later to maintain
the flow of the
argument. The proofs of these
lemmas are gathered in Section \ref{sec:proofs-of-lemmas}.

For $0 \leq \alpha \leq 1$ let $\Omega_{\alpha n}$ denote the set of
configurations
$\Omega_\sigma$ where $\sigma_1 = \alpha n$ and $\sigma_2 = \sigma_3 =
(1-\alpha)n/2$.
In the next step, we show that by losing only a polynomial factor, the
numerator of (\ref{eq:cond1}) can be bounded by the sums of the
probabilities of the
configurations $\Omega_{n/2}$ (the set of configurations on the
boundary $B$ with equal numbers of green and blue vertices), while the
denominator is certainly is
as large as the weight of the configurations in $\Omega_{n/3}$ (the
set of configurations with equal numbers of red, blue and green vertices).
In particular, we want to show that for some constant $C$,

\begin{eqnarray}
\label{eq:using-unimodality}
\frac{\displaystyle\sum_{i \in I}\displaystyle\sum_{x \in B}\
\pi_{\beta_i}(x)}{\displaystyle\sum_{i \in I}\displaystyle\sum_{x \in A}\
\pi_{\beta_i}(x)}    & \leq &
C n  \frac{\displaystyle\sum_{i \in I}\
\pi_{\beta_i}(\Omega_{n/2})}{\displaystyle\sum_{i \in
 I}\
\pi_{\beta_i}(\Omega_{n/3})}
\end{eqnarray}

We use the following lemma, which says that in the simplex, along the line where the
number of red vertices is $n/2$, the distribution at every temperature
has a unique maximum at the configurations where the number of
green vertices is equal to the number of blue vertices.
We define the following function $\overline{\pi}_{i} (x)$ which interpolates the discrete density $\pi_{\beta_i}$ continuously:
$$
\overline{\pi}_{i} (x) =    \frac{\Gamma(n)}{\Gamma(\frac n2)\Gamma(xn)\Gamma\left( \left(\frac 12 - x\right)n\right)} 
\frac{e^{\overline \beta_i n^2 \left( \left(\frac{1}{2}\right)^2+x^2+ \left(
\frac{1}{2}-x \right)^2  \right) }}{\overline Z(\beta_i)}, \\ \ \ \ \ \ \ \  \ \ x \in (0,1/2).
$$
%
\begin{lemma}
\label{diff}
The function $\overline{\pi}_{i}(x)$ has a unique maximum such that $xn$ is integer in the range
$0<x<\frac{1}{2}$ and attains its maximum
at $x=\frac{1}{4}$ for all $i$ such that $\beta_i \le \beta_c$.
\end{lemma}

The proof
appears in Section \ref{sec:proofs-of-lemmas}. This implies the inequality (\ref{eq:using-unimodality}). 
Next, we'll show that $\Phi_S$ is essentially determined by the
conductance of the cut induced at
the highest inverse temperature $\beta_M$.

\begin{lemma} \label{8}
For every inverse temperature $\beta_i \le \beta_c$,
$\frac{\pi_{i}(\Omega_{{n/2}})} {\pi_{i}(\Omega_{{n/3}})} \leq
\frac{\pi_{M}(\Omega_{{n/2}})} {\pi_{M}(\Omega_{{n/3}})} $.
\end{lemma}

\begin{proof}
Note that only the exponential term in $\frac{\pi_{i}(\Omega_{{n/2}
})}
{\pi_{i}(\Omega_{{n/3}})}$ varies with $\beta_i$. Letting $h(n)$ be the ratio of the multinomial terms,
we have
\begin{eqnarray*}
\frac{\pi_{i}(\Omega_{{n/2}})}{\pi_{i}(\Omega_{{n/3}})} \ &=& \
h(n)^{\overline \beta_in^2(\overline H(\frac12,\frac14,\frac14)-\overline H(\frac13,\frac13,\frac13))}\\
\  &=&  \ h(n)e^{\overline \beta_in^2(1/24)} \\
\  &\le & \ h(n)e^{\overline \beta_cn^2(1/24)}
\ =  \
\frac{\pi_{\beta_c}(\Omega_{{n/2}})}{\pi_{\beta_c}(\Omega_{{n/3}})}.
\end{eqnarray*}
 \end{proof}

This implies that the ratio on the right hand side of (\ref{eq:using-unimodality}) can be
bounded as follows
\begin{eqnarray}
\label{eq:critical-beta}
\frac{\displaystyle\sum_{i \in I}\
\pi_{\beta_i}(\Omega_{n/2})}{\displaystyle\sum_{i \in
 I}\
\pi_{\beta_i}(\Omega_{n/3})}   \leq c_6 n
\frac{\pi_{\beta_c}(\Omega_{{n/2}})}{\pi_{\beta_c}(\Omega_{{n/3}})}.
\end{eqnarray}
for some constant $c_6>0$.
There are two final steps to bounding the conductance. First, we
will show that $\frac{\pi_{\beta_c}(\Omega_{{n/2}})}{\pi_{\beta_c}(\Omega_{{n/3}})}$
is exponentially small. Second, we will show that $\Phi \leq
n^{O(1)} \Phi_S$. These facts follow from properties of the stationary
distribution proved in Lemmas \ref{lem:exp} and \ref{lem:within-poly}.

The following lemma demonstrates that there is a critical temperature
at which $\Omega_{n/3}$ and $\Omega_{2n/3}$
both have large weight compared to $\Omega_{n/2}$. Also, the
configurations $\Omega_{n/3}$ have a weight that is at least a
polynomial fraction of the stationary weight of $\Omega$ at $\beta_c$.
\begin{lemma}
\label{lem:exp}

At $\beta_c = \frac{2 ln 2}{n}$,
\begin{enumerate}
        \item[(i)] $\pi_{\beta_c}(\Omega_{{n/3}}) =
\pi_{\beta_c}(\Omega_{{2n/3}}) + o(1).$
        \item[(ii)]
          $\frac{\pi_{\beta_c}(\Omega_{{n/2}})}{\pi_{\beta_c}(\Omega_{{n/3}})}
\leq e^{-\Omega(n)}$
\item[(iii)] $\pi_{\beta_c}(\Omega_{n/3}) \geq
  \frac{\pi_{\beta_c}(\Omega)}{n^2}$
\end{enumerate}
\end{lemma}
The proof of the lemma can be found in Section \ref{sec:proofs-of-lemmas}.
Putting together the bound on $\Phi_S$ from inequality
(\ref{eq:critical-beta}) and part ii) of Lemma \ref{lem:exp}, we obtain that
for some constant $c_7>0,$
\begin{eqnarray*}
\Phi_S \leq e^{-c_7n}
\end{eqnarray*}
Lastly, we show the bound on the conductance $\Phi$. We need the following
lemma, which says that the stationary weight of the configurations on
either side of the cut $S$ are within a polynomial factor.

\begin{lemma}\label{lem:within-poly}
The stationary weight in the tempering chain of the set $S$ is
bounded as $\pi_{st}(S) \leq n^{O(1)} \pi_{st}(S^c)$.
\end{lemma}
\begin{proof}
\begin{eqnarray*}
\pi_{st}(S^c) & = & \frac{1}{M+1} \displaystyle\sum_{i \in I}
  \displaystyle\sum_{x \in \Omega \setminus A}\pi_{\beta_i}(x) \\
& \geq & \frac{1}{M+1}
  \pi_{\beta_c}(\Omega_{2n/3}) \\
\rm{(By \ Lemma \ \ref{lem:exp} \ \textit{(i)} )\ } & \ge & \frac{1}{M+1}
  \pi_{\beta_c}(\Omega_{n/3})\\
\rm{(By \ Lemma \ \ref{lem:exp} \ \textit{(iii)} )\ } & \geq &  \frac{1}{4n^2}
  \frac{1}{M+1}
  \pi_{\beta_c}(\Omega) \\
& \ge & \frac{1}{4n^2} \frac{1}{M+1} \pi_{st}(S) 
\end{eqnarray*}
where the last inequality follows since $\pi_{\beta_i}(\Omega) = 1$.
 \end{proof}
With this lemma in hand, we can bound the conductance of the tempering
Markov chain at the temperature $\beta_c$.
\begin{eqnarray*}
\Phi_{S^c} & = & \frac{F_{S^c}}{C_{S^c}} \\
\rm{(By \ Lemma \ \ref{lem:within-poly}) \ } & \leq & n^{O(1)}
\frac{F_{S^c}}{C_S} \\
\rm{(By \ reversibility) \ } & = &  n^{O(1)} \frac{F_S}{C_S} \\
& \leq & n^{O(1)} e^{-c_7n}.
\end{eqnarray*}

This bounds the conductance since $\Phi \leq
\max(\Phi_S,\Phi_{\overline{S}}) \leq e^{c_4n}$ for some $c_4>0$.
This
completes the proof of Theorem \ref{thm:small-conductance} which implies that simulated tempering is torpidly mixing.
Zheng \cite{Z} has shown that polynomial mixing of the swapping Markov chain
implies polynomial mixing of the tempering chain.
Thus the torpid mixing of simulated tempering implies that the swapping chain for
the mean-field Potts model mixes exponentially torpidly also, for the same distributions, completing the proof of Theorem \ref{thm:slow-potts}.


\subsection{Proofs of Technical Lemmas}
\label{sec:proofs-of-lemmas} 
We present the proofs now of the technical lemmas about the stationary distribution which were used in the previous section. \\

\noindent {\bf   Lemma \ref{diff}.}
The function $\overline{\pi}_{i}(x)$ has a unique maximum such that $xn$ is integer in the range
$0<x<\frac{1}{2}$ and attains its maximum
at $x=\frac{1}{4}$ for all $i$ such that $\beta_i \le \beta_c$. 

\begin{proof}
Recall that we have defined 
$$
\overline{\pi}_{i} (x) =    \frac{\Gamma(n)}{\Gamma(\frac n2)\Gamma(xn)\Gamma\left( \left(\frac 12 - x\right)n\right)} 
\frac{e^{\overline \beta_i n^2 \left( \left(\frac{1}{2}\right)^2+x^2+ \left(
\frac{1}{2}-x \right)^2  \right) }}{\overline Z(\beta_i)}, \\ \ \ \ \ \ \ \  \ \ x \in (0,1/2).
$$

\noindent Neglecting the factors that are not dependent on $x$ we can write the function that we would like to maximize as
\begin{eqnarray*}
& \frac{f(x)}{g(x)} = &
\frac{e^{\overline \beta_i n(x^2+(\frac12-x)^2)}}{\Gamma(xn)\Gamma\left( \left(\frac 12 - x\right)n\right)}
\end{eqnarray*}
and show that the unique maximum is at $x = 1/4$.
To test the sign of the derivative $\left( \frac{f(x)}{g(x)} \right)'$, we
compare the quantities $\frac{f'}{f}$ and
$\frac{g'}{g}$ since both $f$ and $g$ are positive in the interval $(0,1/2)$. It can be verified that $\frac{f'}{f} = \overline \beta_in(4x-1)$ and $$\frac{g'}{g} =n \left( -\sum_{k=1}^\infty \frac{1}{k+xn-1} + \sum_{k=1}^\infty \frac{1}{k+ (\frac 12 - x)n - 1} \right),$$
where we have used that the derivative of the gamma function is given by
\[
\Gamma'(x) = \Gamma(x) \left(-\gamma + \sum_{k=1}^\infty \left( \frac 1k - \frac{1}{k+x-1} \right) \right)
\]
where $\gamma$ is the Euler-Mascheroni constant.
Thus, we verify that there is a stationary point at $x=\frac{1}{4}$ since $\frac{f'}{f} = 0 = \frac{g'}{g}$. We will verify that this is the unique stationary point in (0,1/2).
Using the integral approximation bound $\sum_{k=cn}^{\infty} \frac{1}{k^2} \ge 1/cn$, we have
\begin{align*}
\left( \frac{f'}{f}\right)' & = 4 \overline \beta_i n,  \ \ \ \ \ \ \ \mathrm{and } \\
\left( \frac{g'}{g}\right)' & = n^2 \left( \sum_{k=1}^\infty  \frac{1}{(k+xn-1)^2} +  \sum_{k=1}^\infty  \frac{1}{(k+(\frac 12- x)n-1)^2} \right) \ge n^2 \left( \frac{1}{xn} + \frac{1}{(\frac12 - x)n}\right) = \frac{n}{2x(\frac 12 - x)}
\end{align*}
Therefore, 
\begin{align*}
\left( \frac{g'}{g}\right)' \ge 8n > 4 \overline \beta_c n \ge 4 \overline \beta_i n = \left( \frac{f'}{f}\right)'
\end{align*}
since $\overline \beta_c = \frac{2\log(2)}{n}$.
Since the derivative of $\frac{g'}{g}$ is greater than the derivative of $\frac{f'}{f}$ at each point in $(1/4,1/2)$, there are no stationary points in that interval. A similar argument by symmetry shows there are no stationary points in $(0,1/4)$. \qedhere

%

%

 \end{proof}

\noindent{\bf Lemma \ref{lem:exp} } At $\beta_c = \frac{4\log 2}{n}$
\begin{enumerate}
        \item[(i)] $\pi_{\beta_c}(\Omega_{{n/3}}) =
\pi_{\beta_c}(\Omega_{{2n/3}}) + o(1).$
        \item[(ii)]
$\frac{\pi_{\beta_c}(\Omega_{{n/2}})}{\pi_{\beta_c}(\Omega_{{n/3}})}
\leq e^{-\Omega(n)}$.
\item[(iii)] $\pi_{\beta_c}(\Omega_{n/3}) \geq
  \frac{\pi_{\beta_c}(\Omega)}{n^2}$.
\end{enumerate}
\begin{proof}
(i) \ We solve for $\beta_c$.
                Let
$\pi_{\beta_i}(\Omega_{{n/3}}) \ = \  \pi_{\beta_i}(\Omega_{{2n/3}})
                $. Then,
$$
               {n \choose \frac{2n}{3},\frac{n}{6},\frac{n}{6}}
\frac{e^{\overline \beta_i(\frac{4n^2}{9}+\frac{n^2}{18})}}{Z(\beta_i)}
                        \ = \
                          {n \choose
                        \frac{n}{3},\frac{n}{3},\frac{n}{3} }
                        \frac{e^{\overline \beta_i(n^2/3)}}{Z(\beta_i)}.
$$
\noindent This implies
\begin{eqnarray*}
e^{\overline \beta_i n^2(1/6)}  \ & = & \  \frac{\left(\frac{2n}{3}!\right)\left(\frac{n}
{6}!\right)
                        \left(\frac{n}{6}!\right)}
                        {\left(\frac{n}{3}!\right)\left(\frac{n}{3}!\right)
                        \left(\frac{n}{3}!\right)} \\
 \ & = & \
\frac{\left(\frac{2}{3}\right)^{\frac{2n}{3}}
  \left(\frac{1}{6}\right)^{\frac{n}{3}}}
{\left(\frac{1}{3} \right)^{n}}\left({\frac{1}{\sqrt2}}\right) \left(
                1+O(n^{-1}) \right) \\
 \ & = & \  \frac{2^{\frac{n}3}}{\sqrt2} \left( 1+O(n^{-1}) \right),
\end{eqnarray*}
\noindent
which occurs when $$ \overline \beta_i \ = \ \frac{2 \ln(2)}{n}+\frac{2}{\sqrt2 n^2} \ln
                \left( 1+O(n^{-1}) \right) .$$
\noindent Setting $\beta_c$ to $\frac{4 \ln(2)}{n}$ gives the desired result.

\vskip.1in
\noindent (ii) \ Let $\beta_c = \frac{4\ln(2)}{n}$. Then we have
\begin{eqnarray*}
\frac{\pi_{\overline \beta_c}(\Omega_{n/2})}{\pi_{\overline\beta_c}(\Omega_{n/3})}  \  &
= &
\  \frac{ {n \choose \frac{n}{2},\frac{n}{4},\frac{n}{4}}
                      e^{\overline \beta_c(3n^2/8)}}
{{n \choose \frac{n}{3},\frac{n}{3},\frac{n}{3}}
                       e^{\overline \beta_c(n^2/3)}} \\
 \ & = & \  \frac{\left(\frac{n}{3}!\right)^3}
{\left(\frac{n}{2}!\right)\left(\frac{n}{4}!\right)^2}
e^{\overline \beta_cn^2/24}\\
& = \ & \sqrt{\frac{27}{32}}
\left(\frac{8}{9}\right)^{\frac{n}{2}}e^{{\ln(2)}n/12 }
\left(1+O(n^{-1})\right)\\
& = \ & \sqrt{\frac{27}{32}} \
e^{-\frac{n}{12}\ln\left(\frac
{3^{12}}{2^{13}}\right)}\left(1+O(n^{-1})\right) \leq e^{-\Omega(n)}.
\end{eqnarray*}

\noindent (iii) \ Let $\beta_c = \frac{4\ln(2)}{n}$. Consider
any general point in the simplex of the form $(x,y,1-x-y)$ for
$0 \leq x + y \leq
1$. It can be verified that the function
$$h(x,y) = \frac{f(x,y)}{g(x,y)} = \frac{e^{\overline \beta_c n (x^2+y^2+
    (1-x-y)^2)}}{\Gamma(xn)\Gamma(yn)\Gamma((1-x-y)n)}$$
has a global maximum at $(1/3,1/3)$, i.e. $h(x,y) \leq h(1/3,1/3)$ for all
    $x,y$ such that $0 \leq x+y \leq 1$. This can be shown by checking
    that $h$ is maximized at $(1/3,1/3)$ over all stationary
    points of $h(x,y)$. This implies that $\pi_{\beta_c}(\Omega_{n/3}) \geq
  \frac{\pi_{\beta_c}(\Omega)}{n^2}$.
\end{proof}

\subsection{Polynomial Mixing of Flat-tempering the 3-state
  Potts Model}

The above proof of slow mixing due to a first-order phase transition and the
results of Section \ref{fast} together give the insight that the tempered
distributions should be defined so that the first order discontinuity is not
preserved. We show that the Flat-tempering algorithm with
distributions $\rho_i$ defined below can be used to
efficiently sample from configurations of the 3-state ferromagnetic
mean-field Potts model at any temperature.
The function $f$ in this case is
\begin{eqnarray*}f_i(x) = {n \choose x_1  \ x_2 \ x_3}^{\frac{i-M}{M}}.
\end{eqnarray*}
With $\rho_i$ defined as above,
\begin{align}
\rho_i(\Omega_{\sigma}) \ = \ {n\choose \sigma_1 \sigma_2 \sigma_3} \rho_{i}(x)
\ = \  \frac{1}{Z_i}\left(\rho_{M}(\Omega_{\sigma})\right)^{\frac{i}{M}}. \label{eq:flat-temp-dist}
\end{align}

\begin{theorem}\label{thm:fast-tempering} Let $\beta  = \frac{\mu}{n}$
  for a constant $\mu>0$. Then, for
  some constant $c_8 > 0$ the
  simulated tempering Markov chain $\widehat{P}$ with the
  distributions $\rho_0,
  \ldots, \rho_M$ mixes in time $O(n^{c_8})$.
\end{theorem}
\begin{proof}The proof makes use of the decomposition theorem. The
  strategy is to partition the state space of the tempering chain
  $\Omega_{st}$ into the sets $(\Omega_\sigma,i)$ for each
  equivalence class of configurations $\sigma$ and inverse temperature
  $\beta_i$. To keep the notation simple, which we write the restriction sets as $(\sigma,i)$.
The restriction sets $(\sigma,i)$) are not
  connected by the chain $\widehat{P}$ since it only moves between
  configurations which differ in the spin at exactly one
  vertex. We can get around this technicality by first
  bounding the mixing time of the 2-step chain $\widehat{P}^2$.

The chain $\widehat{P}^2$ mixes in polynomial time, which can be seen by comparison with the chain on $(\Omega_\sigma,i)$ where in each step, the spins at two randomly chosen vertices are exchanged. This follows since the mixing time of this chain is only smaller than the mixing time of the interchange process on the complete graph which is bounded by $O(n \ln n)$ for the complete graph on $n$ vertices (see e.g. \cite[Chapter 14]{AldFil}).


We analyze the projection by comparison to the complete graph on
the states of the projection $\{(\sigma,i)\}$. For
every pair of states $(\sigma,i)$ and $(\sigma',j)$,
we define a path using edges of $\widehat{P}^2$ and show that the
congestion of these paths is at most a polynomial.

Assume without loss of generality that $i \leq j$. Let  $\tau(\sigma,\sigma')$ be a 
sequence of $O(n)$ states that is the set of vertices along a shortest
path using edges of the projection chain in $(\Omega,0)$ from
$\Omega_\sigma$ to $\Omega_{\sigma'}$, not including the endpoints. The path between
$(\sigma,i)$ and $(\sigma',j)$ is defined to be the concatenation of the paths
$((\sigma,i), (\sigma,i-1), \ldots, (\sigma, 0)), \tau(\sigma,\sigma'),$ and 
$((\sigma',0), \ldots, (\sigma',j))$. 
The observations we use to bound the congestion of the paths by a
polynomial is as follows.
\begin{enumerate}[i)]
\item
Let $\sigma_{\max}$ be an equivalence class
of configurations maximizing $\rho_M(\Omega_\sigma)$. By \eqref{eq:flat-temp-dist}, for any $i$,
$$\rho_M^{\frac{i}{M}}(\Omega_{\sigma_{\max}})
\leq Z_i \leq n^{O(1)} \rho_M^{\frac{i}{M}}(\Omega_{\sigma_{\max}}).$$

\item For any edge in the kernel of the Markov chain, the number of
  paths which are routed through it is at most $O(n^4M^2) \leq
 n^{O(1)}$, taking
  into account the possible starting and ending states.
\end{enumerate}
Then, the congestion of the paths can be bounded as follows. We divide
into two cases. The first where an edge corresponds to a change in the
temperature and is of the form $(\sigma,i'),(\sigma,i'-1)$ for some
$i' \leq i$
(or $(\sigma,j'),(\sigma,j'+1)$ for some$j' < j$). The second is an
edge corresponding to a pair of adjacent states at the inverse
temperature $\beta_0$.
\begin{itemize}
\item Assume that the edge is of the form $(\sigma,i'),(\sigma,i'-1)$ for some
$i' \leq i$. By the observations i) and ii) above,
\begin{eqnarray*}
A & \leq & n^{O(1)} \frac{\min\left(
    \frac{\rho_M^{i/M}(\Omega_\sigma)}{\rho_M^{i/M}
    (\Omega_{\sigma_{\max}})},\frac{\rho_M^{j/M}(\Omega_{\sigma'})}
    {\rho_M^{j/M}(\Omega_{\sigma_{\max}})}
    \right)
}
{\min\left(
    \frac{\rho_M^{i'/M}(\Omega_\sigma)}{\rho_M^{i'/M}
    (\Omega_{\sigma_{\max}})},\frac{\rho_M^{(i'-1)/M}(\Omega_{\sigma'})}
    {\rho_M^{(i'-1)/M}(\Omega_{\sigma_{\max}})}
    \right)
}\\
& \leq & n^{O(1)}
    \left[ \frac{\rho_M(\Omega_\sigma)}{\rho_M
    (\Omega_{\sigma_{\max}})}\right]^{\frac{i-i'}{M}} \leq n^{O(1)}
\end{eqnarray*}
The other case is analogous.
\item Suppose that the edge is a pair of adjacent states at
  $\beta_0$. Since for every $\sigma$, $\rho_0(\Omega_\sigma) =
  \Theta(n^{-2})$, we have
\begin{eqnarray*}
A & \leq & n^{O(1)} \frac{\min\left(
    \frac{\rho_M^{i/M}(\Omega_\sigma)}{\rho_M^{i/M}
    (\Omega_{\sigma_{\max}})},\frac{\rho_M^{j/M}(\Omega_{\sigma'})}
    {\rho_M^{j/M}(\Omega_{\sigma_{\max}})}
    \right)}
{n^{-2}} \leq n^{O(1)}
\end{eqnarray*}
\end{itemize}
Finally, by applying the comparison theorem, the polynomial mixing
time of $\widehat{P}^2$ implies that the mixing
time of $\widehat{P}$ is at most a polynomial. This follows since for
any two adjacent states of $\widehat{P}$, the ratio of the stationary
probabilities is at least an inverse polynomial.
Moreover, for any edge of the 1-step chain, there are at most a polynomial
number of possibilities for the other step.
 \end{proof}

\section{Tempering Can Slow Down Fixed Temperature Algorithms}
\label{slower}
We have shown that simulated tempering can mix torpidly. In fact,
tempering can be slower than the fixed temperature algorithm by more
than a polynomial factor.
In this section we show that for the 3-state Potts model, at an inverse temperature $\beta^*$ just above the critical inverse
temperature, on a restricted part of the state space $\Omega$,
simulated tempering can be slower than
the fixed temperature Metropolis chain by an exponential factor. The
idea is that although the mixing time of the Metropolis
chain at $\beta^*$ is exponential, it is bounded by the size of the cut at $\beta^*$, while the
mixing time of the simulated tempering chain can be an exponential
multiplicative factor worse because the conductance of the same cut at
the higher temperatures is much smaller. Intuitively,
on average, the chain is spends even less time mixing on both sides
of the cut at the higher temperatures than at $\beta^*$.
The precise theorem we show is the following. Let us denote by $\Omega_{RGB} = \{x
\in \Omega \ : \ \sigma_1 \geq \sigma_2 \ge \sigma_3 \}$ the subset of the state space for the 3-state Potts model where the number of vertices of the first color dominate the number of the second which in turn dominate the number of vertices of the third color.

\begin{theorem}\label{thm:slowermixing}
Let $\beta^* = \frac{\mu}{n}$ where $\mu > 4\log(2)$. Assume
that the number of distributions for tempering is $M = \Theta(n)$. Then,
there are constants $\delta > 0$ and $\alpha < 0$ (which may depend on
$\beta^*$) such that the
simulated tempering algorithm on $\Omega_{RGB}$ at $\beta^*$ mixes
only after time $\Omega(e^{(-\alpha  + \delta )n})$. The Metropolis
algorithm at temperature $\beta^*$ mixes in time $O(e^{- \alpha n + o(1)})$
\end{theorem}

\subsection{Torpid Mixing of Tempering for $\beta^*> \frac{4\log(2)}{n}$}
We start by proving the first part of the theorem above by showing the
following bound on the conductance of the simulated tempering chain.
Let $\Phi_{RGB}$ denote the conductance of the tempering chain on
$\Omega_{RGB}$ at inverse temperature $\beta^*$.

\begin{proposition}\label{prop:lowerconductance}
Let $\beta^* = \frac{\mu}{n}$ where $\mu > 4\log(2)$.
  Then, there exists $\alpha < 0$ and $\delta >0$ such that
  $\Phi_{RGB} \le e^{(\alpha-\delta) n +
o(n)}.$
\end{proposition}

Define the set
$K_{RGB} = \{\sigma=(\sigma_1, \sigma_2, \sigma_3)$ where
$\sigma_1 \geq \sigma_2 \geq \sigma_3$, \ $\sum_i \sigma_i =
n\}$. Thus $K_{RGB}$ is the set of $\sigma$ corresponding
to the configurations in $\Omega_{RGB}$.
For $\sigma \in K_{RGB}$, the Gibbs distribution is given by
\begin{eqnarray*}
\pi_{\beta_i}(\sigma) = {n \choose \sigma_1 \ \sigma_2 \ \sigma_3}
\frac{e^{\overline \beta_i(\sigma_1^2+ \sigma_2^2+\sigma_3^2)}}{\overline Z_{RGB}(\beta_i)}
\end{eqnarray*}
where $\overline Z_{RGB}(\beta_i)$ is the normalizing constant.

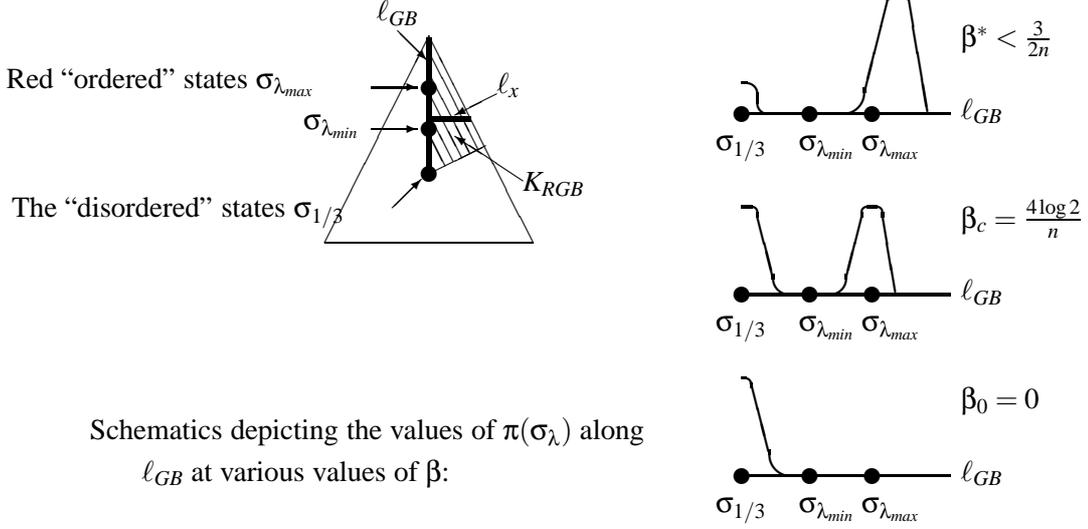
\begin{figure}[t]
\begin{center}
\setlength{\unitlength}{.27in}
\begin{picture}(9,9)

\put(0,5){\line(1,0){4}}
\put(0,5){\line(1,2){2}}
\put(4,5){\line(-1,2){2}}
\put(2,6.33){\circle*{.3}}
\put(2,8){\circle*{.3}}
\put(2,7.2){\circle*{.3}}
\put(2,6.33){\line(2,1){1.09}}
\thicklines
\put(2,6.33){\line(0,1){2.63}}
\put(2,7.4){\line(1,0){.8}}
\put(1.98,6.33){\line(0,1){2.63}}
\put(2,7.38){\line(1,0){.8}}
\put(2.02,6.33){\line(0,1){2.63}}
\put(2,7.42){\line(1,0){.8}}

\thinlines
\put(2,6.7){\line(1,-2){.22}}
\put(2,7.1){\line(1,-2){.31}}
\put(2,7.5){\line(1,-2){.46}}
\put(2,7.9){\line(1,-2){.63}}
\put(2,8.3){\line(1,-2){.78}}
\put(2,8.7){\line(1,-2){.95}}

\put(-6.0,5.5){The ``disordered'' states $\sigma_{1/3}$}
\put(1.3,5.68){\vector(1,1){.56}}
\put(-6.1,8){Red ``ordered'' states $\sigma_{\lambda_{max}}$}
\put(.9,8){\vector(1,0){.88}}
\put(-0.4,7.2){$\sigma_{\lambda_{min}}$}
\put(.9,7.2){\vector(1,0){.88}}
\put(3.3,7.9){$\ell_{x}$} \put(3.3,7.88){\vector(-2,-1){.8}} 
\put(1,9.3){$\ell_{GB}$} \put(1.4,9.1){\vector(1,-1){.56}}
\put(3.8,6){$K_{RGB}$} \put(3.8,6.2){\vector(-3,2){1.3}}

\put(-4.5,1.2){Schematics depicting the values of
  $\pi(\sigma_{{\lambda}})$ along }
\put(-3.5,.4){$\ell_{GB}$ at various values of $\beta$:}

\thicklines \put(8,4){\line(1,0){4}}
\put(8,4){\circle*{.3}}
\put(9.3,4){\circle*{.3}} \put(10.5,4){\circle*{.3}}
\put(7.5,3.3){$\sigma_{1/3}$}
\put(9.1,3.3){$\sigma_{\lambda_{min}}$}
\put(10.3,3.3){$\sigma_{\lambda_{max}}$}
\put(12.2,5.3){$\beta_c=\frac{4\log 2}{n}$}
\put(10.5,5.5){\oval(.4,.4)[t]}
\put(10.7,5.5){\line(1,-6){.24}}
\put(10.3,5.5){\line(-1,-4){.3}}
\put(9.3,4.4){\oval(1.4,.8)[b]}
\put(8.3,5.5){\line(1,-4){.3}}
\put(8,5.5){\oval(.6,.4)[tr]}

\thicklines \put(8,7.5){\line(1,0){4}}
\put(8,7.5){\circle*{.3}}
\put(9.3,7.5){\circle*{.3}} \put(10.5,7.5){\circle*{.3}}
\put(7.5,6.8){$\sigma_{1/3}$}
\put(9.1,6.8){$\sigma_{\lambda_{min}}$}
\put(10.3,6.8){$\sigma_{\lambda_{max}}$}
\put(12.2,8.8){$\beta^*<\frac{3}{2n}$}
\put(8,7.8){\oval(.6,.6)[tr]}
\put(9.1,7.8){\oval(1.6,.6)[bl]}
\put(9.7,8){\oval(1.3,1)[br]}
\put(11,9.7){\oval(.4,.4)[t]}
\put(11.2,9.7){\line(1,-6){.36}}
\put(10.8,9.7){\line(-1,-4){.43}}

\thicklines \put(8,.5){\line(1,0){4}}
\multiput(12.2,.4)(0,3.5){3}{$\ell_{GB}$}
\put(8,.5){\circle*{.3}}
\put(9.3,.5){\circle*{.3}}
\put(10.5,.5){\circle*{.3}}
\put(7.5,-.2){$\sigma_{1/3}$}
\put(9.1,-.2){$\sigma_{\lambda_{min}}$}
\put(10.3,-.2){$\sigma_{\lambda_{max}}$}
\put(12.2,1.8){$\beta_0=0$}

\put(8,2.19){\oval(.36,.4)[tr]}
\put(8.18,2.3){\line(1,-4){.33}}
\put(9.03,1){\oval(1,1)[bl]}

\end{picture}
\end{center}
\caption{The profile of the probability density function over
$K_{RGB}$}\label{profile}
\end{figure}

Denote by
${\ell_{GB}}$, the set of equivalence classes of configurations $\sigma_\lambda = \left(\lambda
n,\frac{(1-\lambda)n}{2}, \frac{(1-\lambda)n} {2}\right)$, for
$\frac{1}{3} \le \lambda \le 1$ i.e., the subset of $K_{RGB}$ with partitions that have an equal number of blue and green vertices (see Figure \ref{profile}).
There exists a constant
$\lambda_{min}$ (which can be found by differentiating the appropriate function), a value of $\lambda$ between the ordered and
disordered modes where $\pi_{\beta^*}(\sigma_{\lambda})$ is
minimized along the line $\ell_{GB}$. Let
$\Omega_{\lambda_{min}}$ be the corresponding set of spin configurations.
Let $\beta_M = \beta^*=\frac{\mu}{n}$ where $\mu$ is a constant such
that $\mu > 4\log(2)$.
Let $A \subseteq \Omega_{RGB}$ be the set of configurations $x$ with
$x_1 \leq \lambda_{min}n$. Let $S =\{(x,i) \
| \ x \in A, \ \beta_0 \le \beta_i \le \beta_c \}$.
Let $B = \{x \in A \ |
\ \exists \ x' \in \Omega_{RGB} \setminus A, \ P(x,x') > 0 \}$ be the boundary of
$A$.
Then, as in \eqref{eq:cond1}, we can bound the conductance of the set $S$ for the tempering chain as
follows.
\begin{eqnarray}\label{eq:conductance}
\Phi_S & \leq &
\frac{\displaystyle\sum_{i = 0}^{M} \displaystyle\sum_{x \in
    B}\pi_{\beta_i}(x)}
{\displaystyle\sum_{i = 0}^{M} \displaystyle\sum_{x \in
    A}\pi_{\beta_i}(x)}
\leq  O(n) \frac{\displaystyle\sum_{i = 0}^{M}
  \pi_{\beta_i}(\Omega_{\lambda_{min}})}
{\displaystyle\sum_{i = 0}^{M} \pi_{\beta_i}(\Omega_{n/3})}
\end{eqnarray}
The second inequality above follows from the fact that the
distribution when restricted to $B$ at every temperature is unimodal and is maximized at
$\Omega_{\lambda_{min}}$.
\begin{lemma}
\label{diff2} Let $\beta^* = \frac{\mu}{n}$ where $\mu > 4\log(2)$.
For $n$ sufficiently large, the continuous function
$\overline{\pi}_{\beta_i}(x) =
\pi_{\beta_i}\left(\lambda_{min}n,(1-\lambda_{min}-x)n,xn
\right)$ has a unique maximum in the range $0 \leq x \leq
1-\lambda_{min}$ at $x=\frac{1-\lambda_{min}}{2}$ for all
$i \in \{1,\ldots,M\}$.
\end{lemma}
The proof of the result above appears at the end of this subsection.
Rewriting the last expression in (\ref{eq:conductance}), we have

\begin{eqnarray}\label{eq:conductance1}
\Phi_S & \leq & O(n)
\frac{\pi_{\beta_M}(\Omega_{\lambda_{min}})}{\pi_{\beta_M}(\Omega_{n/3})}
\frac{\left( \left(
\frac{\pi_{\beta_{M}}(\Omega_{\lambda_{min}})}
{\pi_{\beta_M}(\Omega_{\lambda_{min}})} \right)+
\left(
\frac{\pi_{\beta_{M-1}}
  (\Omega_{\lambda_{min}})}{\pi_{\beta_M}(\Omega_{\lambda_{min}})}
\right)
+ \cdots+
\left(
\frac{\pi_{\beta_{0}}(\Omega_{\lambda_{min}})}
{\pi_{\beta_M}(\Omega_{\lambda_{min}})}
\right)
\right)}
{
\left( \left( \frac{\pi_{\beta_{M}}(\Omega_{n/3})}{\pi_{\beta_M}(\Omega_{n/3})}
\right)
+
\left( \frac{\pi_{\beta_{M-1}}(\Omega_{n/3})}{\pi_{\beta_M}(\Omega_{n/3})}
\right)
+ \cdots+
\left( \frac{\pi_{\beta_{0}}(\Omega_{n/3})}{\pi_{\beta_M}(\Omega_{n/3})}
\right)
\right) }.
\end{eqnarray}

We use the following properties of the stationary distribution to
bound the conductance.
The first fact is that the stationary weight of the
disordered mode conditioned on being at a particular temperature is
non-decreasing as we decrease $\beta$.
\begin{lemma}\label{inc}
For $i \in \{1,\ldots, M\}$, we have that for some $C>1$,
$\pi_{\beta_{i-1}}(\Omega_{n/3}) > C \pi_{\beta_{i}}(\Omega_{n/3})$.
\end{lemma}
\begin{proof}
We have
\begin{align*}
\frac{\pi_{\beta_{i-1}}(\Omega_{n/3})}{\pi_{\beta_{i}}(\Omega_{n/3})} = \frac{\overline{Z}_{RGB}(\beta_i) /\pi_{\beta_{i}}(\Omega_{n/3}) }{\overline{Z}_{RGB}(\beta_{i-1}) / \pi_{\beta_{i-1}}(\Omega_{n/3})} = \frac{\displaystyle\sum_{\sigma \in K_{RGB}} \binom{n}{\sigma_1 \ \sigma_2 \ \sigma_3} e^{\overline  \beta_i (\overline H(\sigma)-\overline H(\sigma_{1/3}))}}{\displaystyle\sum_{\sigma \in K_{RGB}} \binom{n}{\sigma_1 \ \sigma_2 \ \sigma_3} e^{\overline  \beta_{i-1}( \overline H(\sigma)-\overline H(\sigma_{1/3}))}} > C
\end{align*}
for some $C>1$. We obtain the last inequality by arguing as follows.
Since $\overline H(\sigma)$ is minimized at $\sigma_1 = \sigma_2 = \sigma_3 = 1/3$, for each $\sigma \in K_{RGB}$, $\overline H(\sigma) \ge \overline H(\sigma_{1/3})$. In fact, for each $\sigma \ne \sigma_{1/3}$ it is the case that $\overline H(\sigma) - \overline H(\sigma_{1/3}) \ge c n^2$ for some constant $c$, and therefore for $\sigma \ne \sigma_{1/3}$, the ratios $e^{(\overline  \beta_i - \overline  \beta_{i-1})( \overline H(\sigma)-\overline H(\sigma_{1/3}))} > K$ for some constant $K>1$. Moreover, since each of the terms $e^{\beta_i( \overline H(\sigma)-\overline H(\sigma_{1/3}))} > 1$ and $|K_{RGB}| = \Omega(n^2)$, one can find a constant $1< C< K$ such that the inequality above holds.
\end{proof}

%


\noindent Next, we observe that the height of the disordered mode increases
faster than the
height at $\Omega_{\lambda_{min}}$.

\begin{lemma}\label{geo}There is a constant $d>1$ such that $
\frac{\pi_{\beta_{i-1}}(\Omega_{n/3})}
{\pi_{\beta_{i}}(\Omega_{n/3})}
>
d \cdot \frac{\pi_{\beta_{i-1}}(\Omega_{\lambda_{min}})}
{\pi_{\beta_{i}}(\Omega_{\lambda_{min}})}.$
\end{lemma}

\begin{proof}
Expanding the terms shows that
\begin{eqnarray*}
\frac
  {
    \pi_{\beta_{i-1}}(\Omega_{n/3}) /
    \pi_{\beta_{i}}(\Omega_{n/3})}
  {
    \pi_{\beta_{i-1}}(\Omega_{\lambda_{min}}) /
    \pi_{\beta_{i}}(\Omega_{\lambda_{min}})} =
\frac
  {
    \pi_{\beta_{i-1}(\Omega_{n/3})} /
    \pi_{\beta_{i-1}(\Omega_{\lambda_{min}})}}
  {
    \pi_{\beta_{i}(\Omega_{n/3})} /
    \pi_{\beta_{i}(\Omega_{\lambda_{min}})}} =
e^{(\overline  \beta_i-\overline  \beta_{i-1})(\overline H(\Omega_{\lambda_{min}})-\overline H(\Omega_{n/3}))}
\end{eqnarray*}

\noindent Recall that
$\overline  \beta_i-\overline  \beta_{i-1} = O(\frac{1}{nM})$ while
$H(\sigma_{\lambda_{min}})-H(\sigma_{1/3})=\Omega(n^2),$ since $\lambda_{min}$
is a constant. The claim follows since $M = \Theta(n)$.
 \end{proof}
 



By Lemma \ref{geo}, the rate of increase of terms in the series
in the denominator of (4) is at least a constant, $d > 1$, times the rate of
increase of terms in the series in numerator. Combining with Lemma \ref{inc} and using the fact that $M = \Theta(n)$, for some constants $d_3>0$ and $d_2 > 1$, \eqref{eq:conductance1} implies
\begin{eqnarray*}
\Phi_{S}  & \leq  &
O(n)\frac{\pi_{\beta_M}(\Omega_{\lambda_{min}})}{\pi_{\beta_M}(\Omega_{n/3})}
\frac{\left(1+ \left( \frac{d_2}{d}\right) + \ldots + \left(
  \frac{d_2}{d}\right)^{d_3 n}
  \right)}{1+d_2+\cdots+ d_2^{d_3 n}} \\
 & \leq &
O(n)\frac{\pi_{\beta_M}(\Omega_{\lambda_{min}})}
{\pi_{\beta_M}(\Omega_{n/3})}(\min(d_2,d))^{-d_3 n}.
\end{eqnarray*}
Proposition \ref{prop:lowerconductance} follows by setting $\delta =
\frac{d}{3}\ln(\min(d_2,d))$.

\begin{proof}[Proof of Lemma \ref{diff2}]

As in the proof of Lemma \ref{diff}, we define 
$$
\overline{\pi}_{i} (x) =    \frac{\Gamma(n)}{\Gamma(\lambda_{\min}n)\Gamma(xn)\Gamma\left( \left(1-\lambda_{\min} - x\right)n\right)} 
\frac{e^{\overline \beta_i n^2 \left( \lambda_{\min}^2+x^2+ \left(
1-\lambda_{\min}-x \right)^2  \right) }}{\overline Z(\beta_i)}, \\ \ \ \ \ \ \ \  \ \ x \in (0,1-\lambda_{\min}).
$$

\noindent Neglecting the factors that are not dependent on $x$ we can write the function that we would like to maximize as
\begin{eqnarray*}
& \frac{f(x)}{g(x)} = & \frac{e^{\overline \beta_i n^2 \left( x^2+ \left(
1-\lambda_{\min}-x \right)^2  \right) }}{\Gamma(xn)\Gamma\left( \left(1-\lambda_{\min} - x\right)n\right)} 
\end{eqnarray*}
and show that the unique maximum is at $x = (1-\lambda_{\min})/2$.
To test the sign of the derivative $\left( \frac{f(x)}{g(x)} \right)'$, we
compare the quantities $\frac{f'}{f}$ and
$\frac{g'}{g}$ since both $f$ and $g$ are positive in the interval $(0,1-\lambda_{\min})$. It can be verified that $$\frac{f'}{f} = \overline \beta_in(4x-2(1-\lambda_{\min}))$$ and $$\frac{g'}{g} =n \left( -\sum_{k=1}^\infty \frac{1}{k+xn-1} + \sum_{k=1}^\infty \frac{1}{k+ (1-\lambda_{\min}- x)n - 1} \right).$$
Thus, there is a stationary point at $x=\frac{1-\lambda_{\min}}{2}$ since $\frac{f'}{f} = 0 = \frac{g'}{g}$.
We will argue that this is the unique stationary point in $(0,1-\lambda_{\min})$.
Using the integral approximation bound $\sum_{k=cn}^{\infty} \frac{1}{k^2} \ge 1/cn$, we have
\begin{align*}
\left( \frac{f'}{f}\right)' & = 4 \overline \beta_i n,  \ \ \ \ \ \ \ \mathrm{and } \\
\left( \frac{g'}{g}\right)' & = n^2 \left( \sum_{k=1}^\infty  \frac{1}{(k+xn-1)^2} +  \sum_{k=1}^\infty  \frac{1}{(k+(1-\lambda_{min}- x)n-1)^2} \right) \\ 
& \ge n^2 \left( \frac{1}{xn} + \frac{1}{(1-\lambda_{min} - x)n}\right) = \frac{n(1-\lambda_{min})}{x(1-\lambda_{min} - x)}.
\end{align*}
Therefore, for $n$ large enough,
\begin{align*}
\left( \frac{g'}{g}\right)' \ge \frac{4n}{1-\lambda_{min}} > 4 \overline \beta^* n \ge 4 \overline \beta_i n = \left( \frac{f'}{f}\right)'
\end{align*}
since $\overline \beta^* = \frac{\mu}{n}$, a constant.
Since the derivative of $\frac{g'}{g}$ is greater than the derivative of $\frac{f'}{f}$ at each point in $((1-\lambda{min})/2,1-\lambda{min})$, there are no stationary points in that interval. A similar argument by symmetry shows there are no stationary points in $(0,(1-\lambda{min})/2)$. \qedhere

\end{proof}

\subsection{Upper bound for the Metropolis Algorithm on $\Omega_{RGB}$. }
The Metropolis Markov chain on $\Omega$ is known to have exponential
mixing time and the same argument also holds on $\Omega_{RGB}$.
We would now like to derive a good upper bound on this mixing time so that we
can compare it to the bound obtained for the tempering chain.
However, bounding the conductance and applying Theorem
\ref{thm:mixing-conductance} will
not be sufficient as the square of the conductance gives too weak
a bound.
Instead, to obtain the best possible lower bound on the spectral gap of the
Metropolis chain, we appeal to the comparison theorem \cite{DS-C}.  We use
this technique to obtain a tight exponential upper bound for the mixing time.
Let $P$ be the Metropolis chain on $\Omega_{RGB}$ with stationary
distribution $\pi=\pi_{\beta^*}$. Then, the second part of Theorem
\ref{thm:slowermixing} is as follows.
\begin{proposition}
\label{prop:metropolis-rapid}
Let $\beta^* = \frac{\mu}{n}$ where $\mu > 4\log(2)$ and let
$\alpha=\ln\left(\pi_{\beta^*}(\Omega_{\lambda_{min}})/\pi_{\beta^*}
(\Omega_{{1/3}})\right) < 0 $. The Markov chain $P$ mixes in
time $O(e^{-\alpha n+o(1)})$.
\end{proposition}

The idea behind showing the mixing time claimed in Proposition
\ref{prop:metropolis-rapid} is to
define a new distribution $\widetilde{\pi}$ on $\Omega_{RGB}$ by effectively
eliminating the disordered mode. The Metropolis chain $\widetilde{P}$ is
defined on $\Omega_{RGB}$ with stationary distribution
$\widetilde{\pi}$. We will show that the mixing time of
$\widetilde{P}$ is at most a polynomial. The comparison theorem then
gives the required upper bound on
the mixing time of the Metropolis chain $P$ in Proposition
\ref{prop:metropolis-rapid}.
Let $$K := \{\sigma \in \Omega_{RGB} : \sigma_1 < \lambda_{min}n~\hbox{\ and \ }
\pi_{\beta^*}(\Omega_{\sigma}) \ge \pi_{\beta^*}(\Omega_{\lambda_{min}})\}.$$
For $\sigma \in K_{RGB}$ define

\begin{eqnarray*}
\widetilde{\pi}(\Omega_\sigma)  = \left\{
\begin{array}{ll}
\pi_{\beta^*}(\Omega_{\lambda_{min}})(\beta^*)/\widetilde{Z}
 &~\hbox{\ if \ } \sigma \in K  \\
\pi_{\beta^*}(\Omega_\sigma)(\beta^*)/\widetilde{Z} & \hbox{\ otherwise},
\end{array}
\right.\
\end{eqnarray*}
where
$$\widetilde{Z} =  \displaystyle\sum_{\sigma \in
  K}\pi_{\beta^*}(\Omega_{\lambda_{min}}) +
  \displaystyle\sum_{\sigma \in \Omega_{RGB} \setminus K} \pi_{\beta^*}(\Omega_\sigma)$$
  is the normalizing partition function.


For a configuration $x \in \Omega_{RGB}$, we define
$\widetilde{\pi}(x)$ to be uniform
over all the configurations in the same equivalence class, i.e., if
$x$ is in the equivalence class $\sigma$
$$
\widetilde{\pi}(x) = {n \choose \sigma_1 \sigma_2 \sigma_3}^{-1}
\widetilde{\pi}(\Omega_\sigma).
$$

The first step is to show that $\widetilde{P}$, the Metropolis chain on
the flattened distribution, mixes in polynomial time.  This will follow
from an application of the decomposition theorem
\cite{MarR} (see below).
The second step will be to use this bound and the
comparison theorem to bound the mixing time of the chain on the original
unflattened space.
This mixing time of $\widetilde{P}$ will be a lower order term when we
compare it to
the mixing time of $P$, which is exponential. Thus, any polynomial
bound on the mixing rate of $\widetilde{P}$ will suffice.

\begin{theorem}\label{flatdist}
The Markov chain $\widetilde{P}$ with stationary distribution
$\widetilde{\pi}$ mixes in polynomial time.
\end{theorem}

To apply the decomposition theorem here, we partition the space
$\Omega_{RGB}$ according to the equivalence
classes of configurations, i.e. into the space $K_{RGB}$. Informally,
It will be simpler to bound the mixing time of
$Q=\widetilde{P}^2$, the two step transition matrix that allows
moves of length 0, 1 or 2. We can then infer the
polynomial mixing of $\widetilde{P}$ from the polynomial mixing of
$Q$. It is easy to see that $Q$ is polynomially mixing when restricted to
$\Omega_\sigma$, for any $\sigma$, because two-step moves permute
the colors on the vertices without changing the total number of
each and the mixing time can be bounded by that of an interchange process \cite[Chapter 14]{AldFil}. Hence, we focus on showing the bound on projection Markov chain
$\overline{Q}$. We will use the canonical path method.

\begin{theorem}\label{thm:Q-projection-rapid}
The Markov chain $\overline{Q}$ on $K_{RGB}$ is
polynomially mixing.
\end{theorem}

\begin{proof}
For $\sigma \in \Omega_{\sigma}$ and $\tau \in \Omega_\tau$, define the canonical path  $\gamma_{\sigma \tau}$ as follows: Let $\sigma=(t_1, b_1 , g_1 )$ and
$ \tau=(t_2, b_2 , g_2 )$. Assume that $t_1 \ge t_2$. If not, the path from $\sigma$
to $ \tau$ consists of the same vertices as the path from $ \tau$ to $\sigma$
but with all edges directed
oppositely.

We define the canonical path
for $t_1$ odd and $t_2$, the other case only needs a minor technical
modification due to parity issues. Assume (without loss of generality
by the symmetry of the colors blue and green) $b_1 \le g_1$ and
$b_2 \ge g_2$.
The path $\gamma_{\sigma \tau}$ is defined to be
$(t_1,b_1,g_1),(t_1,b_1+1,g_1-1),\ldots,(t_1,\frac{n-t_1-1}{2},
\frac{n-t_1+1}{2}),(t_1-1,\frac{n-t_1+1}{2},\frac{n-t_1+1}{2}),
(t_1-3,\frac{n-t_1+3}{2},\frac{n-t_1+3}{2}),  \ldots
(t_2,\frac{n-t_2}{2},\frac{n-t_2}{2}),\ldots,(t_2,b_2-1,g_2+1),
(t_2,b_2,g_2)$. It can be shown that along the path,
the values of the distribution are unimodal, i.e.,

\begin{lemma}\label{lem:unimodal-canpaths} For each $\sigma =
  (t_1,b_1,g_1), \tau = (t_2,b_2,g_2)
  \in K_{RGB}$, the distribution
  $\widetilde{\pi}$ attains a unique maximum on the path
  $\gamma_{\sigma,\tau}$.
\end{lemma}
We defer the proof till the end of this argument.
Assuming the lemma, the congestion of the paths can be bounded as follows.

\begin{eqnarray*}
A \ & = & \ \max_{(\alpha,\beta)\in E(\overline{Q})}
\left\{{{1}\over{\widetilde{\pi}(\Omega_\alpha)
      {P}^2(\Omega_\alpha,\Omega_\beta)}}
\sum_{\Gamma(\alpha,\beta)} |\gamma_{\sigma \tau}|
   \min(\widetilde{\pi}(\Omega_\sigma),\widetilde{\pi}(\Omega_\tau))
 \right\} \\
& = & \max_{(\alpha,\beta)\in E(\overline{Q})} \left\{{{1}\over{\min(
\widetilde{\pi}(\Omega_\alpha),\widetilde{\pi}(\Omega_\beta))}}
\sum_{\Gamma(\alpha,\beta)}
|\gamma_{\sigma \tau}| \min(\widetilde{\pi}(\Omega_\sigma),\widetilde{\pi}(
\Omega_\tau)) \right\}
\end{eqnarray*}
Since along every canonical path the distribution is unimodal, and the
length of any path
is at most linear in $n$, and there are at most polynomially many
paths $\Gamma(\alpha,\beta)$ using the edge $(\alpha,\beta)$, $A$ is at most a
polynomial in $n$.
\end{proof}

\begin{corollary}The Markov chain $\widetilde{P}$ on $K_{RGB}$ is
  polynomially mixing.
\end{corollary}

\noindent{\bf Proof of Lemma \ref{lem:unimodal-canpaths}: } Let
$\ell_{t}$ denote the set of $\sigma \in K_{RGB}$ such that $\sigma_1 =
t$. Let $\ell_{b=g}$ denote the set consisting of configurations where
the number of green and blue vertices are equal. Since the
space is discrete, because of parity considerations, the
  canonical paths cannot simply go along the line $\ell_{t_1}$, then
  along the line $\ell_{GB}$ and finally along $\ell_{t_2}$, except
  in the case that $t_1$ and $t_2$ are both even. For this case, it is
  sufficient to show that firstly, for all $1/3 \leq t \leq 1$, along
  the lines $\ell_t$, the maximum is at
  the intersection with $\ell_{GB}$ and secondly, along the line
  $\ell_{GB}$, the distribution is unimodal. The observation is that
  the second fact implies that on the portion of the canonical path along
  $\ell_{GB}$, the distribution is either
\begin{enumerate}[i)]
\item non-increasing
\item non-decreasing
\item non-decreasing and then
  non-increasing
\end{enumerate}
but not decreasing and then increasing. Then in any of the three cases
above, it can be verified that there is a unique local maximum along
the path.

In the other cases, when either both $t_1$ and $t_2$ are odd, or one
is odd and the other even, the canonical path makes a ``diagonal''
move to switch parity and we have to argue that the property of being
unimodal is not violated. It turns out that this is implied by the
unimodality of the continuous function $\widetilde{\pi}$ on the lines
$\ell_t$ and $\ell_{GB}$.
We first show that along the lines $\ell_t$ and
$\ell_{\lambda}$ the distribution $\widetilde{\pi}$ is unimodal.

\begin{claim}
\label{d1} Let $\beta^* = \frac{\mu}{n}$ where $\mu > 4\log(2)$ and
let $\ell_t=\{\sigma \in \Omega_{RGB}: \sigma_1=t\}$. Then there
exists a constant $n_0$, such that $\forall n \ge n_0$ the
function $\widetilde{\pi}(\sigma)$ when restricted to $L_t$ is
maximized at $\sigma_2=\sigma_3=\frac{n-t}{2}$ and is
non-increasing as $\sigma_3$ decreases, for all $\lambda_{min} \le t \le 1$.
\end{claim}
The proof follows by the same calculations made in the proof of Lemma \ref{diff2}.

\begin{claim}
\label{d2} Let $\beta^* = \frac{\mu}{n}$ where $\mu > 4\log(2)$. For $n$ sufficiently large, 
$\widetilde{\pi}_{\beta^*}(\Omega_{{\lambda}})$
has a unique maximum $\lambda_{max}$ in $(1/3,\lambda_{min}]$ and is non-increasing on either
side of it.
\end{claim}

\begin{proof}
We examine the continuous extension
$\overline{\pi}$ of the original distribution $\pi$.
\begin{eqnarray*}
\overline{\pi}_{\beta^*} \left(\lambda n,
\frac{(1-\lambda)n}{2},\frac{(1-\lambda)n}{2} \right) =
{n \choose \lambda n,\frac{(1-\lambda)n}{2},\frac{(1-\lambda)n}{2}}
  \frac{e^{\beta^* n^2 \left(\lambda^2+2\left( \frac{1-\lambda}{2}
\right)^2 \right) }}{Z(\beta^*)}
\end{eqnarray*}
Neglecting factors not explicitly dependent on $\lambda$,
asymptotically, we obtain the function
\begin{eqnarray*}
e^{\frac{\beta^*n^2}{2}(3\lambda^2-2\lambda+1)
-\lambda n \ln(\lambda)-(1-\lambda)n \ln\left(\frac{1-\lambda}{2}\right)}
\end{eqnarray*}
The claim can be verified by differentiating it, solving for the
stationary point $\lambda_{max}$, and checking the second derivative.
By construction,
$\widetilde{\pi}_{\beta^*}$ is non-increasing on either side of $\lambda_{max}$
for $\frac{1}{3}~\le~\lambda~\le~1$.

\end{proof}

%
%
%

Finally, along the ``diagonal'' portions of the
path the change in the value of the distribution will be the net
change if we were to move
in a continuous fashion horizontally and then vertically. Since along
both these segments the change will be of the same sign if the
segments on either end are of the same type (increasing or
decreasing), by the two claims above, the net change will be positive
or negative as required by unimodality.\hfill

The Metropolis chain at $\beta^*$ mixes
torpidly, and by the above lemmas we can bound the mixing time. Note
that the proof uses a stronger version of the Comparison Theorem.

To use the comparison theorem to infer a bound on the mixing time of $P$ from
that of $\widetilde{P}$ we need good bounds on the parameters $A$ and $a$.
It turns out that $A$ is the insignificant factor in the mixing time,
rather, $a$ determines
the mixing time of $P$.  In contrast, most previous applications of the
comparison theorem consider chains with identical stationary distributions, so
typically the parameter $a=1$. \\

\noindent{\bf Proof of Proposition \ref{prop:metropolis-rapid}.}  We will use the
refined comparison theorem of Diaconis and Saloff-Coste, Theorem
\ref{thm:refined-compthm}. Note
that the two Markov kernels are
identical,
but their stationary distributions are very different near the
disordered state.  Since the kernels are identical, we can simply define
trivial canonical paths, i.e., when we decompose a step in the unknown
chain $\overline{Q}$ with stationary distribution $\pi_{\beta^*}$ into a
path using steps from the known chain $\overline{Q}$ with distribution
$\widetilde{\pi}$, these paths
all have length 1.  It can be verified that the
Metropolis transition probabilities on the two chains are always
within a polynomial factor of each other and $\max_x (\widetilde{\pi}(x)
/ \pi(x))$ is at most a polynomial since flattening the distribution
has a negligible effect on the partition function.
\begin{claim}Let $\beta^* = \frac{\mu}{n}$ where $\mu > 4\log(2)$. Then,
  $$Z_{RGB}(\beta^*)/n^{O(1)} \leq \widetilde{Z} \leq
  Z_{RBG}(\beta^*).$$
\end{claim}
\begin{proof}The upper bound is easy to see by the definition of
  $\widetilde{Z}$. By the construction of the flattened distribution,
  $\widetilde{Z} \leq Z_{RGB}(\beta^*)$. For the lower bound, we have
\begin{eqnarray*}
\widetilde{Z} & = & Z_{RGB}(\beta^*) \left(\displaystyle\sum_{\sigma \in
  K_1}\pi_{\beta^*}(\Omega_{\lambda_{min}}) +
  \displaystyle\sum_{\sigma \in K_2} \pi_{\beta^*}(\Omega_\sigma)\right) \\
& \geq & Z_{RGB}(\beta^*) \left(
  \displaystyle\sum_{\sigma \in K_2} \pi_{\beta^*}(\Omega_\sigma)\right) \\
& \geq  & Z_{RGB}(\beta^*)/n^{O(1)}
\end{eqnarray*}
The last inequality follows because for $\beta^* > \beta_c$, the stationary
probability on $K_2$ is at least $1/n^{O(1)}$ of the total measure.
\end{proof}
Hence the parameter $A$ is bounded by a
polynomial. Finally, we can compare the largest variation in the distributions
$\pi$ and $\widetilde{\pi}$ to bound $a$. Let $x$ be any configuration
in $\sigma_{1/3}$, any $x_*$ a
configuration in $\sigma_{\lambda_{min}}$ we
have
\begin{eqnarray*}
a \ = \ \frac{\widetilde{\pi}_{\beta^*}(x)}{\pi_{\beta^*}(x)}
=
 \frac{\pi_{\beta^*}(\Omega_{\lambda_{min}}) Z_{RGB}/ \widetilde{Z}}
 {\pi_{\beta^*}(\Omega_{{1/3}})} \geq
\frac{\pi_{\beta^*}(\Omega_{\lambda_{min}})}
 {\pi_{\beta^*}(\Omega_{{1/3}})} \frac{1}{n^{O(1)}}
\end{eqnarray*}
Plugging these bounds into the comparison theorem (Theorem
\ref{thm:refined-compthm}) then implies Proposition \ref{prop:metropolis-rapid}.

\bibliographystyle{plain}
\bibliography{bibliography}

\end{document}